\documentclass[journal]{IEEEtran}

\usepackage{amssymb}
\usepackage{amsmath}
\usepackage{amsfonts}
\usepackage{multirow}
\usepackage{graphicx}
\usepackage{textcomp}
\usepackage{amsthm}
\renewcommand{\vec}[1]{\boldsymbol{#1}}
\usepackage{setspace}
\usepackage{subfigure}
\usepackage{algorithm}
\usepackage{algorithmic}
\usepackage{color}
\usepackage{hyperref}

\allowdisplaybreaks[4]

\ifCLASSINFOpdf
\else
\fi

\newtheorem{thm}{Theorem}
\newtheorem{lem}{Lemma}

\begin{document}

\title{An Online Resource Scheduling for Maximizing Quality-of-Experience in Meta Computing}

\author{Yandi Li,
    Jianxiong Guo,~\IEEEmembership{Member,~IEEE},
    Yupeng Li,~\IEEEmembership{Member,~IEEE},
    Tian Wang,
    and Weijia Jia,~\IEEEmembership{Fellow,~IEEE}
	\thanks{Yandi Li is with the Guangdong Key Lab of AI and Multi-Modal Data Processing, Department of Computer Science, BNU-HKBU United International College, Zhuhai 519087, China. (E-mail: liyandi@uic.edu.cn)
	
	Jianxiong Guo, Tian Wang, and Weijia Jia are with the Advanced Institute of Natural Sciences, Beijing Normal University, Zhuhai 519087, China, and also with the Guangdong Key Lab of AI and Multi-Modal Data Processing, BNU-HKBU United International College, Zhuhai 519087, China. (E-mail: jianxiongguo@bnu.edu.cn; cs\_tianwang@163.com; jiawj@bnu.edu.cn)

        Yupeng Li is with the Department of Interactive Media, Hong Kong Baptist University, Hong Kong. (E-mail: ivanypli@gmail.com)
	
	\textit{(Corresponding author: Jianxiong Guo.)}
	}% <-this 
	\thanks{Manuscript received April xxxx; revised August xxxx.}}

\markboth{Journal of \LaTeX\ Class Files,~Vol.~xx, No.~xx, April~2023}%
{Shell \MakeLowercase{\textit{et al.}}: Bare Demo of IEEEtran.cls for IEEE Journals}

\maketitle

\begin{abstract}
Meta Computing is a new computing paradigm, which aims to solve the problem of computing islands in current edge computing paradigms and integrate all the resources on a network by incorporating cloud, edge, and particularly terminal-end devices. It throws light on solving the problem of lacking computing power. However, at this stage, due to technical limitations, it is impossible to integrate the resources of the whole network. Thus, we create a new meta computing architecture composed of multiple meta computers, each of which integrates the resources in a small-scale network. To make meta computing widely applied in society, the service quality and user experience of meta computing cannot be ignored. Consider a meta computing system providing services for users by scheduling meta computers, how to choose from multiple meta computers to achieve maximum Quality-of-Experience (QoE) with limited budgets especially when the true expected QoE of each meta computer is not known as a priori? The existing studies, however, usually ignore the costs and budgets and barely consider the ubiquitous law of diminishing marginal utility. In this paper, we formulate a resource scheduling problem from the perspective of the multi-armed bandit (MAB). To determine a scheduling strategy that can maximize the total QoE utility under a limited budget, we propose an upper confidence bound (UCB) based algorithm and model the utility of service by using a concave function of total QoE to characterize the marginal utility in the real world. We theoretically upper bound the regret of our proposed algorithm with sublinear growth to the budget. Finally, extensive experiments are conducted, and the results indicate the correctness and effectiveness of our algorithm.
\end{abstract}


\IEEEpeerreviewmaketitle

\section{Introduction}
    %The proposed algorithm is a development of the upper confidence bound (UCB) method with budget constraint imposed. In a classic UCB algorithm, the estimated unknown reward of each arm is characterized with an upper confidence bound index contributing to the decision-making. The index value is composed of the empirical mean and the upper confidence bound of the unknown reward to take into account both the real historical data and the possible underestimate caused by limited sample size. As the index overestimates the reward for a high possibility, the algorithm is endued with exploration ability.

    % 初期，由于技术的限制，无法把整个Internet整合成一个meta computer。因此，local small area network forms a meta computer, 在一个大范围区域，我们有很多个meta computers，每一个都是由其附近网络的资源整合而成。
    % 这里的meta computer和meta computing system定义的有些混乱，我觉得应该是a meta computing system整合了多种资源，而后虚拟化成很多个人meta computers。
    % 因此，现阶段一个可行的解决方案是：A meta computing system consists of multiple meta computers.
    The ever-increasing need for computing powers has been inducing the evolution of computing paradigms, from the Client-Server model to Cloud Computing \cite{josep2010view} \cite{armbrust2010view}, Internet of Things (IoT) \cite{atzori2010internet} \cite{weber2010internet}, and currently Edge Computing \cite{shi2016edge} \cite{premsankar2018edge}. Nowadays, a new paradigm named ``Meta Computing'' \cite{cheng2023meta} has been proposed for more efficient utilization of all accessible computing resources on the network, which can integrate the resources of cloud, edge, and particularly terminal-end IoT devices as a meta computer. By employing blockchain technology, meta computing can guarantee a trusted computing environment for users. The ultimate goal of meta computing is to dynamically transform the Internet and its interconnected resources into a vast pool of computing power in order to empower the latest data-intensive applications that demand strong real-time performance and robust security features. This makes the user seem to be using a supercomputer.

    Unfortunately, in the beginning, due to technical limitations, it is challenging to integrate all resources over the entire Internet into one meta computer (MC). However, we think that it is achievable to integrate all resources on a small-scale nearby network, e.g., resources from a college or a big company. Thus, a feasible early prototype for meta computing is a meta computing system (MCS) consisting of multiple MCs, each of which integrates all resources from a small area nearby through a fast Local Area Network. As a whole, the MCS also integrates resources in a large area and can provide better computing services for users. In this case, an MCS can be compared to a high-performance CPU, and multiple MCs in it are equivalent to multiple CPU cores.

    In our MCS, since it comprises multiple local MCs, a system manager in the MCS identifies tasks and allocates an appropriate MC for these tasks. %of the same category. 
    However, each MC is equipped with different hardware devices and network conditions. This leads to MCs normally having different Quality-of-Service (QoS) and resulting in diverse experiences of service to users \cite{fizza2021qoe}. As the scarcity of computing resources gets alleviated by developments in meta computing, online computing services focus more on user experience which is usually characterized by Quality-of-Experience (QoE) beside computing power \cite{mahmud2019quality} \cite{saadi2020iot} \cite{sodhro2019quality} \cite{abusafia2022quality} \cite{laghari2012toward}. The QoE refers to the level of satisfaction or frustration that the user experiences when using a particular service, which depends on the service type, QoS, and user’s context \cite{aazam2019fog} \cite{chen2014qos} \cite{uthansakul2019estimating} \cite{aazam2017cloud}, accordingly relying on the allocated MC for meta computing. Also, the operation of an MCS is heavily concerned with the planned budget. Thus, it is desirable for an MCS to designate an MC presenting higher QoE per expense with budget concerns \cite{tran2013qoe} \cite{buyya2019fog}. However, in practice, the expected QoE related to an MC is an unknown priori for the MCS, which forces us to seek an online algorithm to deal with the randomness of services provided by the MCS. In addition, in the real world, as the total QoE of the service increases, the additional value obtained from each incremental improvement in QoE will decrease, resulting in the diminishing marginal utility for the total QoE \cite{bao2017prediction}, which yet is neglected in the existing researches about QoE.

    In this paper, we consider a scenario, in which an MCS with a budget constraint and a limited user base provides related computing services to its users for a certain kind of frequently operated tasks, such as industrial Internet of Things analysis or intelligent video surveillance.
    %This paper considers a scenario where an MCS, with a budget constraint and a limited user base for a certain type of frequently operated task such as industrial IoT analysis or intelligent video surveillance, offers related computing services to its users in a certain district.     
    %Users pay monthly fees and send task requests to the MCS, while the cost of utilizing MCs needs to be covered by the MCS. Accordingly, the budget of the MCS depends on the monthly fees from users. As we all know, better QoE can increase the likelihood that users will stay. 
    The MCS prefers to choose the MC with better QoE and lower resource costs. Nevertheless, for the users served by the same MC, the received QoE is assumed to follow the same but unknown distribution due to the fluctuation in hardware resources and users’ context. Therefore, it is crucial for the MCS to effectively explore real QoE-cost ratios of those MCs and find out the optimal one. Although many studies have been conducted to optimize QoE of service by adopting online learning methods against unknown QoE priori \cite{tran2013qoe} \cite{jiang2017pytheas} \cite{boldrini2018mumab} \cite{lu2019automating} \cite{zhou2020human} \cite{zhu2021multi} \cite{lu2022improving} \cite{xu2023resource}, most of them do not take into account the cost or budget, nor do they consider the general economic principle of diminishing marginal utility for QoE. Therefore, we aim to design an efficient online algorithm that can maximize the total diminishing marginal utility within a limited budget.

    There are three main challenges to designing such an online algorithm. The first is to seek a mechanism that can achieve an optimal trade-off between exploration and exploitation on the premise of unknown QoE. The second challenge is to maximize the utility in a budgeted setting. Lastly, due to the budgeted setting as well as the non-linear nature of the concave utility function, the sublinear theoretical guarantee of the difference between the optimal method and the proposed method, known as the \textit{regret}, can be hard to derive. To address the above challenges, we first formulate the problem as a sequential decision problem, which can be modeled by a multi-armed bandit (MAB) with budget constraints. Then, we propose an algorithm based on upper confidence bound (UCB) to tackle the optimal arm-pulling problem for maximum utility concerning costs and budgets. Specifically, we adopt the UCB index instead of empirical mean reward in order to better explore the potential for the high real expected reward of an arm. When analyzing the regret, the properties of the concave function are melted into the inference for a tight upper bound with logarithmic growth with respect to the budget. The main contributions are summarized as follows:
    \begin{itemize}
        \item We introduce a budget-constrained resource scheduling problem within a novel computing paradigm called \textit{meta computing} for the next generation of computing, where different local meta computers correspond to different unknown QoE distributions followed by users. The problem complies with the advanced requirements for meta computing to focus on the QoE of users.
        \item Unlike existing works, we consider the diminishing marginal utility in meta computing so as to model the case in reality. Thus, the utility of an MCS is depicted as a concave function of the total QoE of service users, aligned with the general principle in economics.
        \item We propose an efficient budgeted MAB algorithm based on the UCB mechanism and greedy method to maximize the cumulative utility under a limited budget. Besides, we theoretically upper bound the regret with logarithmic growth to the budget.
        \item We conduct extensive simulations and demonstrate the superior performance of the proposed algorithm over several classic baselines.
    \end{itemize}

    \textbf{Organization.} We first summarize the works related to this paper in Section \ref{sec2}. We then introduce our meta computer system and formulate the problem in Section \ref{sec3}. Next, in Section \ref{sec4} and Section \ref{sec5}, the design of our algorithm and the corresponding theoretical analysis are presented in detail, respectively. We carry out the performance evaluation in Section \ref{sec6} and conclude the paper in Section \ref{sec7}.

\section{Related Work}\label{sec2}
    In this section, we review the related works from three aspects. We first briefly introduce the development of computing paradigms. Then we investigate the recent studies on the problems of online resource scheduling and online service selection that is of similar area, with concerns for the QoE of users. Lastly, we look into the works of underpinned bandit algorithms with budget constraints.

    \textbf{Computing Paradigm.}
    The Client-Server model is a type of architecture proposed for distributed applications with a long history from 1964 \cite{witt1965ibm}, where a server provides resources or services to the clients who make the requests. Subsequently started the Desktop Computing era from the first generation of personal computers in the 1970s. In the 1990s, there appeared Grid Computing that integrates and utilizes idle computing powers such as the high-speed Internet, computers, databases, and sensors of various devices to provide high-performance computing services \cite{foster2003grid}. In 1997, the Cloud Computing paradigm was proposed \cite{chellappa1997intermediaries}, which aims to offer flexible and personalized computing services to businesses and individuals by consolidating computing tasks into a high-performance computing cluster with a virtualized data center \cite{josep2010view} \cite{armbrust2010view}. In 2005, the concept of IoT clarified \cite{strategy2005internet}. IoT computing establishes a connection between the physical world and the cyber world, allowing IoT terminal equipment to continuously collect real-time data. Driven by the real-time requirement as well as security and privacy concerns, Edge Computing was proposed, which aims to make the most of computing resources that are located close to or nearby the network edge, and only rely on remote clouds for data processing tasks when it is absolutely necessary \cite{shi2016edge} \cite{premsankar2018edge}. Recently, Meta Computing has been proposed for more efficient utilization of all accessible computing resources on the network \cite{cheng2023meta}.
    
    \textbf{QoE-Aware Online Resource Scheduling / Service Selection.} As online services tend to be user-centered oriented, QoE management has been introduced and studied for resource scheduling and service selection. Tran \textit{et al.} \cite{tran2013qoe} propose a QoE-based two-layer architecture of a content distribution network, of which the purpose is to choose a proper server to make users as satisfied as possible. They add a server-selection layer besides a routing layer and formulate the server selection as a MAB problem. Then UCB1 algorithm \cite{auer2002finite} is adopted to optimize the selection decision according to users' QoE feedback. Jiang \textit{et al.} \cite{jiang2017pytheas} further consider the divergence among application sessions when improving the QoE. They propose a grouping mechanism for sessions based on features like locations and then utilize a discounted UCB algorithm \cite{garivier2008upper} for QoE optimization on each group. Boldrini \textit{et al.} \cite{boldrini2018mumab} propose a new MAB model QoE maximization problem, where the model adds another action of measuring rather than pulling an arm only and allows actions spanning for multiple rounds in order to accommodate real-world scenarios. Zhou \textit{et al.} \cite{zhou2020human} incorporate the QoE-related context information in channel allocation scenarios for 5G network by encoding it as a stochastic contextual MAB problem, and they solve it with a generalized UCB1 method. There are also studies using MAB approaches to select deep neural network models on edge devices for optimal QoE \cite{lu2019automating} \cite{lu2022improving}. However, these works do not take into account limited budgets. Xu \textit{et al.} \cite{xu2023resource} consider the network selection problem for enlarging QoE under resource budget constraints where the network state information is unknown and dynamic to IoT devices.
    
    \textbf{MAB with Budget Constraints.} In the last decade, MAB algorithms with budget constraints have been extensively studied, where the main objective is to maximize the cumulative reward within limited budgets. Tran-Thanh \textit{et al.} \cite{tran2012knapsack} first propose a budgeted UCB-based MAB model named KUBE and theoretically upper bound the regret with sublinear growth $\mathcal{O}(\log(B))$, where $B$ is the total budget. Ding \textit{et al.} \cite{ding2013multi} extend the scenario to unknown costs. On this basis, Xia \textit{et al.} \cite{xia2015thompson} adopt Thompson sampling and prove the regret bound of $\mathcal{O}(\log(B))$. Then, Xia \textit{et al.} \cite{xia2016budgeted} study the problem of combinatorial MAB with budget constraints, where the player pulls a fixed number of arms in each round, and they propose a policy named MRCB and prove the sublinear regret. Das \textit{et al.} \cite{das2022budgeted} further release the assumptions about fixed budget per round or fixed number of arms per round and propose two UCB-based algorithms which outperform existing works. 
    %However, all these works simply consider linear reward functions, while the concave reward functions are pervasive in reality.

    % 第一：以前的问题都是xxx，我们考虑的是Meta computing。之前从未有过。体系结构上是创新的。
    % 第二：diminishing marginal utility，之前的online scheduling没有。
    % 第三：xxx
    \textbf{Discussion.} Our work is distinct from the existing works. First, the previous works mainly solve QoE-aware online resource scheduling problems with CDN or wireless network service selection contexts \cite{tran2013qoe} \cite{jiang2017pytheas} \cite{boldrini2018mumab} \cite{zhou2020human} \cite{lu2019automating} \cite{lu2022improving} \cite{xu2023resource}, while we consider scheduling optimization in a future-oriented computing paradigm, Meta Computing \cite{cheng2023meta}, which pioneers the new architectural design problem. Second, few of those take into account limited budgets \cite{xu2023resource}, and none involves diminishing marginal utility. Our problem considers both constraints to meet the more complex settings in the real world. Third, from the algorithmic perspective, the pervasive concave reward functions, rather than simple linear reward functions in the existing works \cite{tran2012knapsack} \cite{ding2013multi} \cite{xia2015thompson} \cite{das2022budgeted}, are first combined with budgeted MAB in this paper. Therefore, we have not only done pioneering work for the new computing paradigm of meta computing, but also integrated the concave function as an objective into the MAB model, which has made theoretical contributions to the field of online algorithms.

\section{System Model \& Problem Formulation}\label{sec3}
    In this section, we first introduce the MCS and inner resource scheduling models. Then, we abstract the problem of resource scheduling in the MCS and formulate it rigorously from both offline and online learning perspectives. The frequently used notations are summarized in Table \ref{table0}.

    \begin{figure}[!t]
	\centering
        \includegraphics[width=\linewidth]{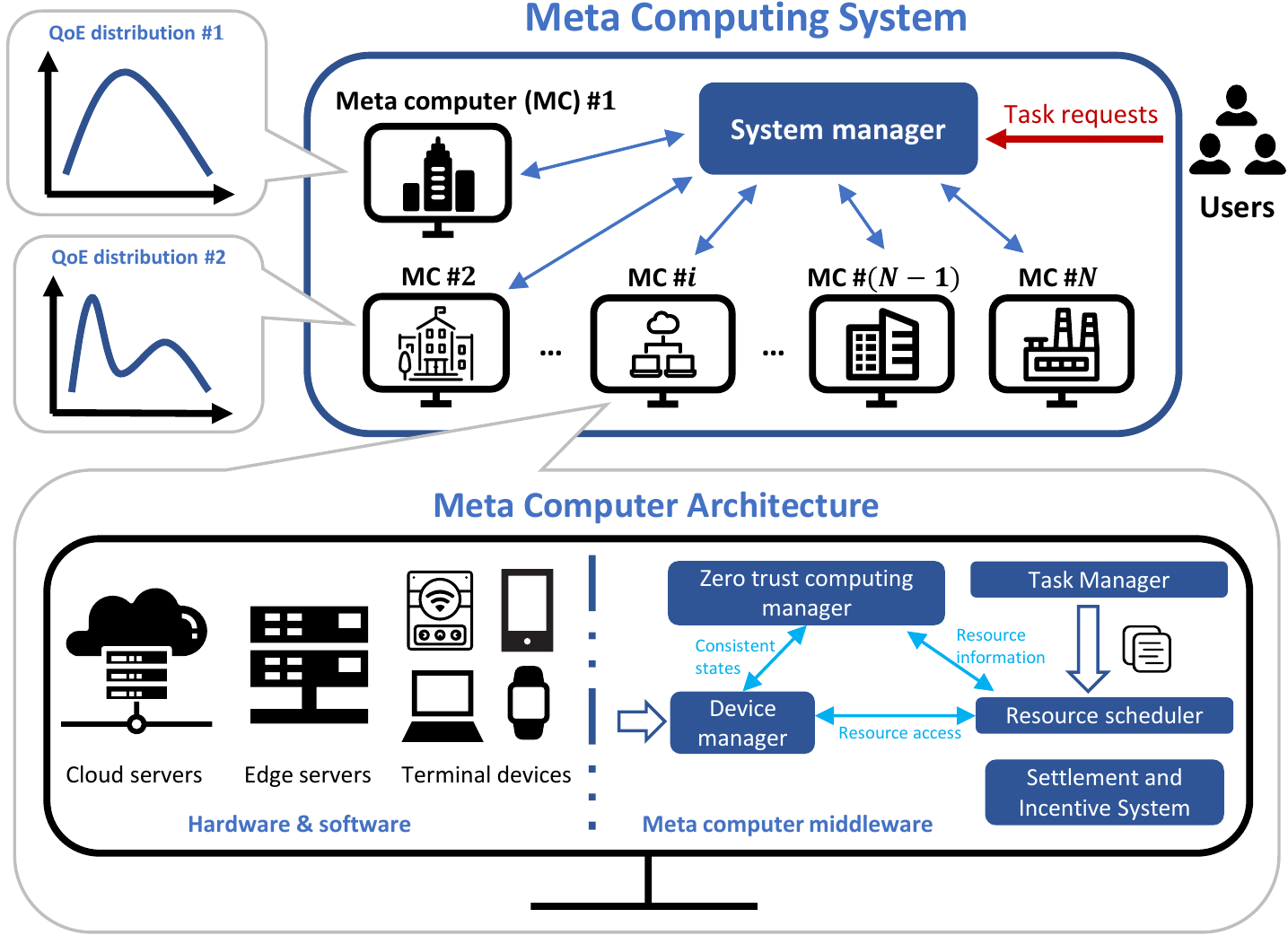}
	\caption{A meta computing system: it consists of $N$ meta computers, each of which integrates all nearby resources including clouds, edges, and terminal devices. The meta computer architecture was proposed in \cite{cheng2023meta}.}
	\label{fig0}
    \end{figure}
    
    \subsection{Meta Computing System}
        The architecture of the meta computing system (MCS) is shown in Fig. \ref{fig0}. In our setting, we consider an MCS consisting of one system manager and $N$ local meta computers (MCs), which is denoted by $\mathcal{N}=\{1,2,\cdots,i,\cdots,N\}$. Here, each MC integrates all resources, including clouds, edges, and terminal devices, in a small-scale network. Users of this system submit their computing tasks to the system manager. After parsing the received tasks, the system manager is designed to dispatch tasks of the same type to an appropriate MC. Subsequently, the task manager and resource scheduler modules within the selected MC further process the tasks for execution. 
        
        Shown as Fig. \ref{fig0}, the structure of an MC encompasses several modules including a device manager, a zero-trust computing manager, and other auxiliary modules \cite{cheng2023meta}. The device manager plays a crucial role in dynamically sensing and mapping heterogeneous online resources into a resource space, which the resource scheduler can utilize for efficient resource scheduling. The zero-trust computing manager establishes a trusted environment for resource sharing and employs blockchain techniques to synchronize the states among system components. Upon receiving computing tasks from the system manager of MCS, the task manager in the MC preprocesses the tasks based on various constraints and then forwards the decomposed task specifications to the resource scheduler for resource allocation. The incurred costs of using the hardware of MC should be paid by the MCS to the hardware resource providers, which may vary depending on task specifications and MCs. %Users pay monthly fees and send task requests to the MCS, while the cost of utilizing MCs needs to be covered by the MCS. Accordingly, the budget of the MCS depends on the monthly fees from users. %Payments for services are made by service receivers along with task requests, which may vary depending on task specifications and MCs.

    \subsection{Resource Scheduling Model}
        For a certain type of routine task and a fixed set of users $\mathcal{M}=\{1,2,\cdots,j,\cdots,M\}$, the MCS selects one MC among total $N$ available MCs following some strategy to provide computing service. In this case,  the users pay monthly fees and send task requests to the MCS, while the cost of utilizing MCs needs to be covered by the MCS. %As is known to all, better QoE can increase the likelihood that users will stay. 
        The timeline is discretized into time slots $\mathcal{T}=\{1,2,\cdots,t,\cdots\}$, and each of them corresponds to a round in which the MCS can complete a cycle including deploying an MC, processing the tasks and collecting the feedback from users. Utilizing the $i$-th MC in a time slot $t$ will produce the cost $M\cdot c_i$, where $c_i\in \left[c_{min}, c_{max}\right]$ is the corresponding unit cost of serving users by $i$-th MC, $c_{min}=\min_{i\in\mathcal{N}}\{c_i\}$ and $c_{max}=\max_{i\in \mathcal{N}}\{c_i\}$. We denote the budget for this type of task as $B$ which is closely dependent on the monthly fees of the corresponding users. The cost incurred by using MC for each round will be deducted from the constrained total budget $B$ until the remaining budget is insufficient to begin a new round. Moreover, the QoE of each MC follows a distribution that is unknown to the MCS. We assume that every MC possesses the adequate capacity and acceptable latency performance for all the users in $\mathcal{M}$. The discrepancies of facilities lead to different QoE among MCs. The users served by the same MC receive the same expected QoE with variations. To characterize the diminishing marginal effect on utility in the real world, for each MC, we introduce a concave utility function $g(\cdot)$ of the summation of QoE retrieved from all users \cite{vickrey1945measuring} \cite{viscusi1990utility}. From the perspective of the MCS, the objective is to maximize the accumulated utility under budget $B$.

    \begin{table}[!t]
	\renewcommand{\arraystretch}{1}
	\caption{The frequently used notations summarization}
	\label{table0}
    %\small
	\centering
	\begin{tabular}{|p{0.2\linewidth}|p{0.7\linewidth}|} 
		\hline
		Notation& Description\\
		\hline
        $\mathcal{N}$, $\mathcal{M}$ & The sets of MCs in an MCS and subscribing users, respectively\\
        $N$, $M$ & The numbers of MCs and users of an MCS, respectively\\
        $\mathcal{T}$ & The set of rounds or time slots\\
        $T_i(t)$ & The number of times the QoE of the $i$-th MC has been revealed after round t\\
        $A(t)$ & The action at round $t$\\
        $B$ & The total initial budget of the MCS\\
        $B_t$ & The remaining budget after round $t$\\
        $\tau(B)$ & The last round regarding the budget\\
        $c_i$ & The unit cost of using the $i$-th MC\\
        $c_{max}$, $c_{min}$ & The maximum and minimum values of unit costs among $\mathcal{N}$\\
        $u_{ij}(t)$ & The QoE received by the user $j$ from the MCS when using MC $i$ in round $t$\\
        $U_{i}$ & The true expectation of QoE of the $i$-th MC\\
        $u_{max}$, $u_{min}$ & The maximum and minimum values of the QoE among $\mathcal{N}$\\
        $\hat{u}_i(t)$ & The empirical mean of QoE for MC $i$ in round $t$\\
        $\varepsilon_{i}(t)$ & The upper confidence bound of $\hat{u}_i(t)$\\
        $\bar{u}_i(t)$ & The UCB index of MC $i$ after round $t$\\
        $I^*$ & The optimal choice of MC from $\mathcal{N}$\\
		\hline
	\end{tabular}
    \end{table}

    \subsection{Problem Formulation}
        Let $\mathcal{N}=\{1,2,\cdots,i,\cdots,N\}$ denote the set of MCs and $\mathcal{M}=\{1,2,\cdots,j,\cdots,M\}$ denote the set of users, where we have $\lvert\mathcal{N}\rvert=N$ and $\lvert\mathcal{M}\rvert=M$. Let $u_{ij}(t)$ denote the QoE received by the user $j$ from the MCS when using MC $i$ in round $t$, which can be regarded as a random variable with an expectation of $U_i$ and is assumed to be independent of each other. In other words, the expected QoE provided by MC $i$ is $U_i$, thus the expected total QoE among all $M$ users can be expressed as $M\cdot U_i$. Note that the QoE of the users who are served by the same MC share the same distribution. For the convenience of analysis, we further constrain $u_{ij}(t)$ as $1\leq u_{min}\leq u_{ij}(t)\leq u_{max}\leq 2$, which is also applicable to other kinds of reward with different range by scaling. In round $t$, the MCS selects a certain MC $A(t)=i$ where $i\in\mathcal{N}$ and the MCS pays $M\cdot c_{A(t)}$ to rent the resources of MC $A(t)$. Define the remaining budget after round $t$ as 
        \begin{equation*}
            B_t=B-\sum_{t^{'}\leq t}M\cdot c_{A(t')},
        \end{equation*}
        then the last round regarding the budget, denoted by $\tau(B)$, which can be characterized as $B_{\tau(B)}\geq 0$ and $B_{\tau(B)+1}\leq 0$. Here, we assume that the total budget $B$ is adequate such that $\tau(B)\geq N$. In order to make the cumulative utility maximized within the constrained budget, an appropriate strategy for choosing an MC is needed for this problem. When the expected QoE of each MC is known, it is an unbounded knapsack problem and the optimal solution is to select the MC with the maximum utility-cost ratio. Let $T_i(t)$ denote the number of times the QoE of the $i$-th MC has been revealed after the round $t$. Note that when the $i$-th MC has been selected in round $t$, the corresponding QoE would be revealed for $M$ times from all the users. The unbounded knapsack problem can be formulated below.
        \begin{align}
        	\max\quad & \mathbb{E}\left[\sum_{i=1}^{N}T_i\left(\tau(B)\right)g\left(M\cdot U_i\right)\right] \notag\\
        	s.\;t.\quad & \sum_{i=1}^{N}T_i\left(\tau(B)\right)Mc_i\leq B. \notag
        	%& d_{ij} \leq d_c, \;\forall i \in \mathcal{N}, \;\forall j\\
        \end{align}
        Here, $g(\cdot)$ is an increasing and concave function. We use it to describe the diminishing marginal utility of total QoE. The marginal income of QoE is decreasing, which conforms to the low of economics, thus avoiding blindly improving QoE.

        However, in our situation, instead of just exploiting complete prior knowledge for optimization, the expected QoE is unknown and needs to be explored through real-world service activities, which cannot be solved by offline algorithms for knapsack problems. Hence, we adopt MAB methods to strike a balance between exploration and exploitation. Specifically, for an MCS, each local MC can be regarded as an arm in MAB, while the served users share a common QoE distribution with the same expectation as a reward in MAB. The algorithm continuously revises its policy according to the updated feedback data collected by actually serving users. The feedback of QoE $u_{ij}(t)$ of MC $i$ is a series of independent and identical distributed (I.I.D.) random variables with respect to $i$, i.e. with the same expectation $U_i$. The expected QoE $U_i$ of MC $i$ can be estimated by the empirical mean $\hat{u}_i(t)$ as defined later. For convenience, we take a specific increasing and concave function $g(\cdot)=\log(\cdot)$ as an example to analyze, where $\log(\cdot)$ in this paper represents the natural logarithmic function with the base $e$. Thus the optimization problem can be re-formulated as follows.
        \begin{align}
        	\max\quad & \mathbb{E}\left[\sum_{t=1}^{\tau(B)}\log\left(\sum_{j=1}^{M}u_{A(t)j}\left(t\right)\right)\right] \notag\\
        	s.\;t.\quad & \sum_{t=1}^{\tau(B)}Mc_{A(t)}\leq B. \notag
        	%& d_{ij} \leq d_c, \;\forall i \in \mathcal{N}, \;\forall j\\
        \end{align}
        Thus, this is the objective of our resource scheduling problem in meta computing.
        
\section{Algorithm Design}\label{sec4}
    One of the core challenges for a MAB algorithm is designing an efficient mechanism for efficiently balancing exploration and exploitation. The proposed algorithm embeds the mechanism of exploration and exploitation into the UCB index. Since there is a limited budget in our scenario, the MC with a higher utility-cost ratio is preferable, thus the algorithm always selects the MC with the largest index-related utility-cost ratio and a permitted cost. We define the empirical mean $\hat{u}_i(t)$ and UCB index $\bar{u}_i(t)$ with respect to arm $i$ after round $t$ as follows:
    \begin{align}
        \hat{u}_i(t) &=\begin{cases}\frac{\hat{u}_i(t-1)T_i(t-1)M+\sum_{j=1}^{M}u_{ij}(t)}{T_i(t-1)M+M}, & A(t)=i\\
        \hat{u}_i(t-1), & A(t)\neq i\end{cases} \notag\\
        \bar{u}_i(t) &=\hat{u}_i(t)+\varepsilon_{i}(t), \notag\\
        \varepsilon_{i}(t)&=\sqrt{{2\log\left(t\right)}/{(T_i(t)M)}}, \notag
    \end{align}
    where $\varepsilon_{i}(t)$ is the upper confidence bound and $T_i(t)$ is the number of times the QoE of MC $i$ has been revealed after the round $t$, which can be formally defined as $T_i(t)=\sum_{k=1}^{t}\mathbb{I}\{A(k)=i\}$. The increment of $T_i(t)$ can only occur if MC $i$ has been selected in round $t$:
    \begin{align}
        T_i(t)=\begin{cases}T_i(t-1)+1,& A(t)=i\\ T_i(t-1),& A(t)\neq i\end{cases}. \notag
    \end{align}

    The online algorithm based on UCB indexes is shown in Algorithm \ref{BUCB}, which is constituted of \textit{Initialization Stage} and \textit{Arm Selection Stage}. The basic framework is of the greedy method with UCB indexes to solve the exploration-exploitation trade-off. In the first stage, shown as line 3 to 9 in Algorithm \ref{BUCB}, each MC is chosen exactly once to construct the initial estimate $\hat{u}_i(t)$. In the second stage, shown as line 11 to 20 in Algorithm \ref{BUCB}, at round $t$, it selects the cost-permitted MC $I$ with the relatively largest utility-cost ratio under the remaining budget left after round $t-1$ as
    \begin{equation*}
        I\leftarrow\mathop{\arg\max}_{\{i:Mc_i\leq B_{t-1}\}} \frac{\log\left(\bar{u}_i(t-1)\right)}{c_i},
    \end{equation*}
    where we use the corresponding utility-cost ratio instead of the unknown expected QoE $U_i$. Then, we observe the value of $u_{Ij}(t)$ for each user $j\in\mathcal{M}$, re-compute the empirical mean $\hat{u}_I(t)$, and get the remaining budget $B_t$. Finally, we update the upper confidence bound $\varepsilon_i(t)$ and get the updated UCB index $\bar{u}_i(t)$ for each MC $i\in\mathcal{N}$, which will be used to select the best MC in the next round.       
    
    \begin{algorithm}[!t]
    	\caption{\text{Budgeted-UCB}}
    	\begin{algorithmic}[1]\label{BUCB}
    		\renewcommand{\algorithmicrequire}{\textbf{Input:}}
    		\renewcommand{\algorithmicensure}{\textbf{Output:}}
    		\REQUIRE $\mathcal{N}$, $\mathcal{M}$, $B$, $\vec{c}=\{c_1,\cdots,c_N\}$;
    		%\ENSURE $\hat{\vec{u}}(\tau(B))$, $\vec{c}$;
            \STATE Initialize $t=0$, $N=\lvert \mathcal{N} \rvert$, $M=\lvert \mathcal{M} \rvert$, $T_i=0$, $B_t=B$, $c_{min}=\min \vec{c}$, $\forall i\in\mathcal{N}$; \label{step1}
            \STATE // \textit{Initialization Stage}
            \FOR{each $i\in\{1,\cdots,N\}$}
                \STATE $t++$;
                \STATE Choose MC $i$ with the cost $c_{i}$;
                \STATE Collect the feedback QoE $u_{ij}(t)$ for each $j\in\mathcal{M}$;
                \STATE $\hat{u}_{i}=\frac{T_{i}M\hat{u}_{i}+\sum_{j=1}^{M}u_{ij}(t)}{T_{i}M+M}$;
                \STATE $B_t=B_{t-1}-Mc_i$, $T_i++$;
            \ENDFOR \label{step6}
            \STATE // \textit{Arm Selection Stage}
            \WHILE{$B_t\geq M\cdot c_{min}$}\label{step7}
                \STATE $t++$;
                \STATE Choose the MC $I=\mathop{\arg\max}_{\{i:Mc_i\leq B_{t-1}\}} \frac{\log\left(\bar{u}_i\right)}{c_i}$; \label{step11}
                \STATE Collect the feedback QoE $u_{Ij}(t)$ for each $j\in\mathcal{M}$;\label{step12}
                \STATE $\hat{u}_{I}=\frac{T_{I}M\hat{u}_{I}+\sum_{j=1}^{M}u_{Ij}(t)}{T_{I}M+M}$;
                \STATE $B_t=B_{t-1}-Mc_I$, $T_I++$;\label{step13}
                \FOR{each $i\in\{1,\cdots,N\}$}
                    \STATE $\varepsilon_{i}=\sqrt{\frac{2\log\left(t\right)}{T_{i}M}}$, $\bar{u}_i=\hat{u}_i+\varepsilon_{i}$;\label{step9}
                \ENDFOR
            \ENDWHILE\label{step14}
    		%\RETURN 
    	\end{algorithmic}
    \end{algorithm}

\section{Theoretical Analysis}\label{sec5}
For online algorithms, the performance is primarily characterized by \textit{regret}. In this section, we first analyze the expression of the regret in our problem, then we further deduce the upper bound of the regret and prove its logarithmic complexity as $\mathcal{O}(\log(B))$ theoretically. The analytic process of regret is dependent on the execution process shown as Algorithm \ref{BUCB}, which can validate the effectiveness of our proposed mechanism.
\subsection{Regret Analysis}
    The regret of a bandit algorithm refers to the reward difference between the optimal algorithm and the current one. In fact, maximizing the cumulative reward of an algorithm is essentially equivalent to minimizing the corresponding regret. Let $R_B(A)$ denote the accumulative reward of algorithm $A$ under the budget $B$, which here means the total utility throughout the rounds before running out of the budget. Then we can define the regret of $A$ under the budget $B$ as $Reg_B(A)$. That is
    \begin{equation}\label{eq8}
        Reg_B\left(A\right)=R_B\left(A^*\right)-R_B\left(A\right),
    \end{equation}
    where $A^*$ represents the optimal algorithm with known expected utility values of all the MCs. In this situation, the optimal algorithm is obviously choosing the MC with the greatest utility-cost ratio among all the applicable MCs in each round. We denote the optimal choice of MC as $I^*$, thus we have
    \begin{equation}\label{eq9}
        I^*=\mathop{\arg\max}_{i\in\mathcal{N}} \frac{\log\left(U_i\right)}{c_i}.
    \end{equation}
    Thus we can give the expressions of $R_B(A^*)$ and $R_B(A)$ as below. Note that the number of rounds to select $I^*$ with respect to the optimal algorithm should be at most $\lfloor B/(M\cdot c_{I^*})\rfloor$. Thus, we have
    \begin{align}  
        &R_B\left(A^*\right)\leq\sum_{t=1}^{\lfloor B/(Mc_{I^*})\rfloor+1 }\log\left(\sum_{j=1}^{M}u_{I^*j}\left(t\right)\right),\label{eq10}\\
        &R_B\left(A\right)=\sum_{t=1}^{\tau(B)}\log\left(\sum_{j=1}^{M}u_{A(t)j}\left(t\right)\right).\label{eq11}
    \end{align}
    Therefore, based on Eqn. (\ref{eq8}) (\ref{eq9}) (\ref{eq10}) (\ref{eq11}), the upper bound of the regret can be further expressed as follows.
    \begin{align}
        Reg_B\left(A\right)\leq&\sum_{t=1}^{\lfloor B/(Mc_{I^*})\rfloor+1}\log\left(\sum_{j=1}^{M}u_{I^*j}\left(t\right)\right) - 
 \notag\\
        &\sum_{t=1}^{\tau(B)}\log\left(\sum_{j=1}^{M}u_{A(t)j}\left(t\right)\right). \notag
    \end{align}
    As aforementioned, the QoE $u_{ij}(t)$ of an MC follows an unknown distribution as a random variable. In addition, the budget constraint leads to a variable total number of rounds $\tau(B)$ based on $u_{ij}(t)$ and the specific algorithm exploited, which also grants randomness to $\tau(B)$. Accordingly, we consider the expectation of $Reg_B\left(A\right)$ due to its stochastic setting and re-defined the problem as:
        \begin{align}
        	\min\quad & \mathbb{E}_{\tau(B),\{A(t)\}}\left[Reg_B\left(A\right)\right] \notag\\
        	s.\;t.\quad & \sum_{t=1}^{\tau(B)}Mc_{A(t)}\leq B. \notag
        	%& d_{ij} \leq d_c, \;\forall i \in \mathcal{N}, \;\forall j\\
        \end{align}
    The objective is to design an algorithm with sublinear regret in reference to the total budget $B$, as even the worst algorithm can result in linear regret. Here, both $\tau(B)$ and $i$ are concerned with the expectation of $Reg_B\left(A\right)$, which brings in extra challenges to analysis.
    
    For clarity of analysis, we define the difference of costs and the smallest difference of utility-cost ratios regarding the suboptimal MC $i$ and the optimal MC $I^*$:
    \begin{align}
        &\delta_i=c_i-c_{I^*}, \notag\\
        &\Delta_i=\frac{\log\left(U_{I^*}\right)}{c_{I^*}}-\frac{\log\left(U_{i}\right)}{c_{i}}, \notag\\ %\quad i\neq I^*\\
        &\Delta_{min}=\min_{i}\Delta_i. \notag
    \end{align}

\subsection{Regret Upper Bound}
    The total regret $Reg_A\left(B\right)$ is closely related to the total number of times arm $i$ has been chosen, $T_i(\tau(B))$, as well as the total number of rounds, $\tau(B)$. Once an MC $i$ is chosen, the difference of utility-cost ratios $\Delta_i$ of this round becomes fixed. In other words, the regret can be determined by knowing $T_i(\tau(B))$ given $\tau(B)$. Before the detailed regret analysis, we introduce some auxiliary lemmas.

    \begin{lem}[Mean value theorem] \label{mvt}
        If function $f(\cdot)$ is continuous and differentiable on $(a,b)$, then there exists $\xi\in (a,b)$ such that:
        \begin{equation}
            f'(\xi)=\frac{f(b)-f(a)}{b-a}, \notag
        \end{equation}
        where $f'(\cdot)$ denotes the derivative of $f(\cdot)$.
    \end{lem}
    \begin{proof}
        The proof of this well-known mean value theorem can be found, for example, in \cite{sahoo1998mean}.
    \end{proof}

    \begin{lem}[Hoeffding's inequality]\label{hoeffding}
        Given a sequence of random variables, denoted as $X_1,X_2,\cdots,X_n$, such that $a_i\leq X_i\leq b_i$ almost surely. Denote the sum of these random variables as $S_n=X_1+\cdots+X_n$. Then, for all $\varepsilon>0$, we have:
        \begin{align}
            \mathbb{P}\left(S_n-\mathbb{E}\left[S_n\right]\geq \varepsilon\right) \leq \exp\left(-\frac{2\varepsilon^2}{\sum_{i=1}^{n}(b_i-a_i)^2}\right),\notag \\
            \mathbb{P}\left(S_n-\mathbb{E}\left[S_n\right]\geq \varepsilon\right) \leq \exp\left(-\frac{2\varepsilon^2}{\sum_{i=1}^{n}(b_i-a_i)^2}\right). \notag
        \end{align}
    \end{lem}
    \begin{proof}
        The proof can be found in \cite{hoeffding1963probability}\cite{boucheron2013concentration}.
    \end{proof}

    \begin{lem}\label{log_ineq}
        Consider the natural logarithmic function $\log(\cdot)$. If $x,y\geq1$ and $x>y$, then $\log(x)-\log(y)<x-y$. Also, a variant of this lemma claims that if $\alpha\geq1$ and $\beta>0$, then $\log(\alpha+\beta)<\log(\alpha)+\beta$.
    \end{lem}
    \begin{proof}
        According to the mean value theorem in Lemma \ref{mvt}, we have:
        \begin{equation}
            \frac{\log(x)-\log(y)}{x-y}=\left(\log(\xi)\right)'=\frac{1}{\xi},\quad \xi\in(y,x). \notag
        \end{equation}
        Since $\xi>1$, $1/\xi<1$, thus $\log(x)-\log(y)<x-y$. Based on this conclusion, we let $x=\alpha+\beta$ and $y=\alpha$, then we have $\log(\alpha+\beta)-\log(\alpha)<(\alpha+\beta)-\alpha$ which is essentially the variant lemma of $\log(\alpha+\beta)<\log(\alpha)+\beta$.
    \end{proof}

    \begin{lem}\label{1to3}
        Let $\bar{\mu}_{i,s_i}$ denotes the empirical mean utility of MC $i$ when $T_i(t)=s_i$ given a fixed time $t$. We can define the following events as 
        \begin{align}
            \Phi &= \left\{\frac{\log\left(\bar{\mu}_{i,s_i}\right)}{c_i} \geq \frac{\log\left(\bar{\mu}_{I^*,s_{I^*}}\right)}{c_{I^*}}\right\}, \notag\\
            \Gamma &= \left\{\frac{\log\left(\hat{u}_{I^*,s_{I^*}}+\varepsilon_{I^*}(t)\right)}{c_{I^*}} \leq \frac{\log\left(U_{I^*}\right)}{c_{I^*}}\right\}, \notag\\
            \Lambda &= \left\{\frac{\log\left(U_i+\varepsilon_{i}(t)\right)}{c_i} \leq \frac{\log\left(\hat{u}_{i,s_i}\right)}{c_i}\right\}, \notag\\
            \Omega &= \left\{\frac{\log\left(U_{I^*}\right)}{c_{I^*}} \leq \frac{\log\left(U_i+2\varepsilon_{i}(t)\right)}{c_i}\right\}. \notag
        \end{align}
        If $\Phi$ holds, at least one of the three events $\Gamma, \Lambda, \Omega$ holds. Moreover, if taking $\ell=\frac{8\log(t)}{Mc_{i}^2\Delta_{i}^2}$ and $T_i(t)\geq \ell$ where $\Delta_{i:i\neq I^*}=\frac{\log(U_{I^*})}{c_{I^*}}-\frac{\log(U_{i})}{c_{i}}$, then we have $\mathbb{P}(\Omega)=0$.
    \end{lem}   
    
\begin{proof}
    For the first proposition of Lemma \ref{1to3}, we use contradiction to prove it. Assume none of the events $\Gamma, \Lambda, \Omega$ hold on condition that $\Phi$ holds, i.e. $\Phi, \neg \Gamma, \neg \Lambda, \neg \Omega$ hold. Then from $\neg \Gamma$ we have $\hat{u}_{I^*,s_{I^*}}+\varepsilon_{I^*}(t)>U_{I^*}$, and from $\neg \Lambda$ we have $\hat{u}_{i,s_i}+\varepsilon_{s_i}(t)<U_i+2\varepsilon_{s_i}(t)$. Combining $\Phi$, the following relationship can be obtained:
    \begin{align}
        \frac{\log\left(U_{I^*}\right)}{c_{I^*}}&< \frac{\log\left(\hat{u}_{I^*,s_{I^*}}+\varepsilon_{I^*}(t)\right)}{c_{I^*}} \notag\\
        &\leq \frac{\log\left(\hat{u}_{i,s_i}+\varepsilon_{i}(t)\right)}{c_i} < \frac{\log\left(U_i+2\varepsilon_{i}(t)\right)}{c_i}, \notag
    \end{align}
    which contradicts $\neg \Omega$, thus the first proposition holds. Thus, it is impossible that none of the events $\Gamma, \Lambda, \Omega$ holds.

    For the second proposition of Lemma \ref{1to3}, We first introduce an inequality by applying the variant of Lemma \ref{log_ineq}: 
    \begin{equation}\label{eq25}
        \log(U_i+2\varepsilon_{i}(t))<\log(U_i)+2\varepsilon_{i}(t).
    \end{equation}
    Then, let $\ell=\frac{8\log(t)}{Mc_{i}^2\Delta_{i}^2}$, we have:
    \begin{align}
        \frac{\log\left(U_{I^*}\right)}{c_{I^*}}&- \frac{\log\left(U_i+2\varepsilon_{i}(t)\right)}{c_i} \notag\\
        &>\frac{\log(U_{I^*})}{c_{I^*}} - \frac{\log(U_{i})}{c_{i}} - \frac{2\varepsilon_{i}(t)}{c_i} \notag\\ 
        &=\frac{\log(U_{I^*})}{c_{I^*}} - \frac{\log(U_{i})}{c_{i}} - \frac{2\sqrt{\frac{2\log\left(t\right)}{T_i(t)M}}}{c_i}\notag\\
        &\geq\frac{\log(U_{I^*})}{c_{I^*}} - \frac{\log(U_{i})}{c_{i}} - \frac{
        2\sqrt{\frac{2\log\left(t\right)}{\ell M}}}{c_i} \notag\\
        &=\Delta_{i} - \Delta_{i}\notag\\
        &=0.\notag
    \end{align}
    Therefore, we have
    \begin{equation}
        \frac{\log\left(U_{I^*}\right)}{c_{I^*}} > \frac{\log\left(U_i+2\varepsilon_{i}(t)\right)}{c_i}, \notag
    \end{equation}
    and then $\mathbb{P}(\Omega)=0$, where the first inequality above is obtained from Eqn. (\ref{eq25}).
\end{proof}
        
    \begin{lem}\label{conf_bound}
       Given a fixed time $t$ and $1\leq u_{ij}(t)\leq 2$, then we have
       \begin{align}
           &\mathbb{P}\left(\hat{u}_{i,s_i}-U_i\geq \varepsilon_{i}(t)\right) \leq t^{-4}, \notag\\
           &\mathbb{P}\left(\hat{u}_{i,s_i}-U_i\leq -\varepsilon_{i}(t)\right) \leq t^{-4}. \notag
       \end{align}
    \end{lem}
    
    \begin{proof}
        We prove the first inequality, and the proof of the second is similar.
        \begin{align}
            \mathbb{P}&\left(\hat{u}_{i,s_i}-U_i\geq \varepsilon_{i}(t)\right) \notag\\
            &=\mathbb{P}\left(s_{i}M\hat{u}_{i,s_i}-s_{i}MU_{i}\geq s_{i}M\varepsilon_{i}(t)\right) \notag\\
            &\leq \exp\left(-\frac{2\left(s_{i}M\varepsilon_{i}(t)\right)^2}{\sum_{i=1}^{s_{i}M}(2-1)^2}\right) \notag\\ 
            &= \exp\left(-{2\left(s_{i}M\sqrt{\frac{2\log\left(t\right)}{s_{i}M}}\right)^2}/{(s_{i}M)}\right) \notag\\
            &= t^{-4}. \notag
        \end{align}
        The first inequality is obtained from Hoeffding's inequality shown as Lemma \ref{hoeffding}. By pre-setting the value pattern of $\varepsilon_{i}(t)$, we manage to determine the confidence level with $t^{-4}$, which diminishes over time.
    \end{proof}

    In order to bound $\mathbb{E}\left[Reg_B\left(A\right)\right]$, we first bound $\mathbb{E}[T_i(\tau(B))\vert\tau(B)]$ and $\mathbb{E}[\tau(B)]$. 

    \begin{thm}\label{Ti_bound}
        Given the total number of rounds $\tau(B)$ within the budget $B$, considering any MC $i\in\mathcal{N}$, the expectation of $T_i(\tau(B))$ can be upper bounded as:
        \begin{equation}
            \mathbb{E}\left[T_i(\tau(B))\vert \tau(B)\right]\leq\frac{8\log(\tau(B))}{Mc_i^2\Delta_i^2} + 1 + \frac{\pi^2}{3}. \notag
        \end{equation}
    \end{thm}

    \begin{proof}
        We first upper bound $T_i(\tau(B))$, then take the expectation of it to get the result. To simplify the proof, we sometimes omit the notation of \textit{conditional to $\tau\left(B\right)$}, i.e. all the probabilities are considered to be conditional on a given $\tau\left(B\right)$. Then, we have $T_i\left(\tau\left(B\right)\right)=$
        \begin{flalign}
            &= 1+\sum_{t=N+1}^{\tau(B)}\mathbb{I}\{A\left(t\right)=i\}\notag\\
            &=1+\sum_{t=N+1}^{\tau(B)}\mathbb{I}\{A\left(t\right)=i,\ T_i(t-1)\geq \ell\}\notag\\
            &\qquad+\sum_{t=N+1}^{\tau(B)}\mathbb{I}\{A\left(t\right)=i,\ T_i(t-1)<\ell\}\notag\\
            &\leq \ell + \sum_{t=N+1}^{\tau(B)}\mathbb{I}\{A\left(t\right)=i,\ T_i(t-1)\geq \ell\}\notag\\
            &\leq \ell + \sum_{t=N+1}^{\tau(B)}\mathbb{I}\left\{\textstyle\frac{\log\left(\bar{\mu}_i(t-1)\right)}{c_i} \geq \frac{\log\left(\bar{\mu}_{I^*}(t-1)\right)}{c_{I^*}},T_i(t-1)\geq \ell\right\}\notag\\
            %&\leq \ell + \sum_{t=N+1}^{\tau(B)}\mathbb{I}\left\{\frac{\log\left(\bar{\mu}_i(t-1)\right)}{c_i} \geq \frac{\log\left(\bar{\mu}_{I^*}(t-1)\right)}{c_{I^*}},\right.\notag\\
            %&\qquad \bigg. T_i(t-1)\geq \ell\bigg\}\notag\\
            &\leq \ell + \sum_{t=N+1}^{\tau(B)}\mathbb{I}\left\{\textstyle \bigcup\limits_{s_{I^*}=1}^{t-1}\bigcup\limits_{s_i=\ell}^{t-1}\left\{\frac{\log\left(\bar{\mu}_{i,s_i}\right)}{c_i} \geq \frac{\log\left(\bar{\mu}_{I^*,s_{I^*}}\right)}{c_{I^*}}\right\}\right\}\label{eq29} \\
            &\leq \ell + \sum_{t=N+1}^{\tau(B)}\sum_{s_{I^*}=1}^{t-1}\sum_{s_i=\ell}^{t-1}\mathbb{I}\left\{\textstyle \frac{\log\left(\bar{\mu}_{i,s_i}\right)}{c_i} \geq \frac{\log\left(\bar{\mu}_{I^*,s_{I^*}}\right)}{c_{I^*}}\right\}\label{eq30}\\
            &\leq \ell + \sum_{t=N+1}^{\tau(B)}\sum_{s_{I^*}=1}^{t-1}\sum_{s_i=\ell}^{t-1}\left(\mathbb{I}\{\Gamma\} + \mathbb{I}\{\Lambda\} + \mathbb{I}\{\Omega\}\right).\label{eq31}
        \end{flalign}
        Here, the Ineqn. (\ref{eq29}) is because of the implication relationship between events, the Ineqn. (\ref{eq30}) can be obtained by taking the union bound, and the Ineqn. (\ref{eq31}) is the direct induction from Lemma \ref{1to3}.
        
        Next, letting $\ell=\frac{8\log(t)}{Mc_{i}^2\Delta_{i}^2}$, we upper bound $\mathbb{E}\left[T_i(\tau(B))\vert \tau(B)\right]=$      
        \begin{align}
            %&=\sum_{t=1}^{\tau(B)}\mathbb{P}\left(A(t)=i\vert \tau(B)\right) \\
            &\leq\ell +  \mathbb{E}\left[\sum_{t=N+1}^{\tau(B)}\sum_{s_{I^*}=1}^{t-1}\sum_{s_i=\ell}^{t-1}\left(\mathbb{I}\{\Gamma\} + \mathbb{I}\{\Lambda\} + \mathbb{I}\{\Omega\}\right)\right] \notag \\
            &= \ell + \sum_{t=N+1}^{\tau(B)}\sum_{s_{I^*}=1}^{t-1}\sum_{s_i=\ell}^{t-1}\left(\mathbb{P}(\Gamma) + \mathbb{P}(\Lambda) + 0\right) \label{p=0}\\
            &=\ell + \sum_{t=N+1}^{\tau(B)}\sum_{s_{I^*}=1}^{t-1}\sum_{s_i=\ell}^{t-1}\left[\mathbb{P}\left(\frac{\log\left(\hat{u}_{I^*,s_{I^*}}+\varepsilon_{I^*}(t)\right)}{c_{I^*}}\right. \right.\notag\\
            &\qquad\left.\left.\leq \frac{\log\left(U_{I^*}\right)}{c_{I^*}}\right)+\mathbb{P}\left(\frac{\log\left(U_i+\varepsilon_{i}(t)\right)}{c_i} \leq \frac{\log\left(\hat{u}_{i,s_i}\right)}{c_i}\right)\right] \notag\\   
            &=\ell + \sum_{t=N+1}^{\tau(B)}\sum_{s=1}^{t-1}\sum_{s_i=\ell}^{t-1}\left(\mathbb{P}\left(\hat{u}_{I^*,s_{I^*}}+\varepsilon_{I^*}(t) \leq U_{I^*}\right)\right. \notag\\
            &\qquad + \left.\mathbb{P}\left(U_i+\varepsilon_{i}(t) \leq \hat{u}_{i,s_i}\right)\right) \notag\\ 
            &\leq \ell + \sum_{t=N+1}^{\tau(B)}\sum_{s_{I^*}=1}^{t-1}\sum_{s_i=\ell}^{t-1}\left(t^{-4} + t^{-4}\right)\label{eq34}\\            
            &\leq \left\lceil \frac{8\log(\tau(B))}{Mc_i^2\Delta_i^2} \right\rceil + \sum_{t=N+1}^{\tau(B)}\sum_{s_{I^*}=1}^{t}\sum_{s_i=1}^{t}2t^{-4} \notag\\
            &\leq \frac{8\log(\tau(B))}{Mc_i^2\Delta_i^2} + 1 + \sum_{t=1}^{\infty}(2t^{-2}) \notag \\
            &\leq \frac{8\log(\tau(B))}{Mc_i^2\Delta_i^2} + 1 + \frac{\pi^{2}}{3}\notag
        \end{align}
        In Eqn. (\ref{p=0}), we obtain $\mathbb{P}(\Omega)=0$ from the proof of Lemma \ref{1to3}. In Eqn. (\ref{eq34}), it is derived from Lemma \ref{conf_bound}.
    \end{proof}

    \begin{lem} \label{con-e}
        Consider two discrete random variables $X$ and $Y$ with joint probability mass function $P(x,y)$. If $f(x,y)$ is a function of $x$ and $y$, the expectation of $f(x,y)$ can be represented as:
        \begin{equation}
            \mathbb{E}_{X,Y}\left[f\right] = \mathbb{E}_{X}\mathbb{E}_{Y\vert X}\left[f\right], \notag
        \end{equation}
        where $\mathbb{E}_{Y\vert X}$ indicates the conditional expectation given $X$.
    \end{lem}

    \begin{proof}
        We use a property of conditional probability to transform the summation for proof:
        \begin{align}
            \mathbb{E}_{X,Y}\left[f\right]&=\sum_{x}\sum_{y} P(x,y)f(x,y) \notag\\
            &=\sum_{x}\sum_{y} P(x)P(y\vert x)f(x,y) \notag\\
            &=\sum_{x}P(x) \sum_{y}P(y\vert x)f(x,y) \notag\\
            &=\mathbb{E}_{X}\mathbb{E}_{Y\vert X}\left[f\right]. \notag
        \end{align}
        Thus, this lemma has been proven.
    \end{proof}

    \begin{thm}\label{tau_bound}
        Under the constrained budget $B$, the expectation of the total number of rounds $\tau(B)$ can be lower bounded by $\mathbb{E}\left[\tau(B)\right] \geq$
        \begin{equation}
            \frac{B}{Mc_{I^*}} - \frac{c_{min}}{c_{I^*}}
            -\sum_{\delta_i>0}\frac{\delta_i}{c_{I^*}}\left(\frac{8\log\left(\frac{B}{Mc_{min}}\right)}{Mc_i^2\Delta_i^2} + 1 + \frac{\pi^{2}}{3}\right).\notag
        \end{equation}
    \end{thm}
    \begin{proof}
        According to our algorithm, after the final round, no MC $i$ can be taken with the remaining budget $B_{\tau(B)}=B-\sum_{t=1}^{\tau(B)}Mc_{A(t)}$. Thus, we have
        \begin{equation}
            B - M\cdot c_{min} \leq \sum_{t=1}^{\tau(B)}M\cdot c_{A(t)} \notag
        \end{equation}
        almost surely. Take the expectation on both sides of this inequality, then we adopt Lemma \ref{con-e} in order to tackle the expectation with respect to $i$ first under the condition of $\tau(B)$. Thus, we have
        \begin{align}
            &\frac{B}{M} - c_{min} \leq \mathbb{E}_{\tau(B),A(t)}\left[\sum_{t=1}^{\tau(B)}c_{A(t)}\right] \notag \\
            &=\mathbb{E}_{\tau(B)}\mathbb{E}_{\{A(t)\}\vert \tau(B)}\left[\sum_{t=1}^{\tau(B)}c_{A(t)}\right] \notag \\
            &=\mathbb{E}_{\tau(B)}\left[\sum_{t=1}^{\tau(B)}\mathbb{E}_{A(t)}\left[c_{A(t)}\right]\right] \notag \\
            &= \mathbb{E}_{\tau(B)}\left[\sum_{t=1}^{\tau(B)}\sum_{i=1}^{N}\mathbb{P}\left(A(t)=i\right)c_{i}\right] \notag \\
            &= \mathbb{E}_{\tau(B)}\left[\sum_{t=1}^{\tau(B)}\left(\left(\sum_{i=1}^{N}\mathbb{P}\left(A(t)=i\right)c_{i}\right) - c_{I^*} + c_{I^*}\right)\right] \notag \\
            &= \mathbb{E}_{\tau(B)}\left[\sum_{t=1}^{\tau(B)}\left(\left(\sum_{i=1}^{N}\mathbb{P}\left(A(t)=i\right)(c_{i}-c_{I^*})\right) + c_{I^*}\right)\right] \notag \\
            &= \mathbb{E}_{\tau(B)}\left[\sum_{t=1}^{\tau(B)}\left(c_{I^*} + \sum_{i=1}^{N}\delta_i\mathbb{P}\left(A(t)=i\right)\right)\right] \notag \\
            &= \mathbb{E}_{\tau(B)}\left[\sum_{t=1}^{\tau(B)}c_{I^*}\right] + \mathbb{E}_{\tau(B)}\left[\sum_{t=1}^{\tau(B)}\sum_{i=1}^{N}\delta_i\mathbb{P}\left(A(t)=i\right)\right] \notag \\
            &\leq\mathbb{E}_{\tau(B)}\left[\tau(B)\right]c_{I^*} + \mathbb{E}_{\tau(B)}\left[\sum_{t=1}^{\tau(B)}\sum_{\delta_i>0}\delta_i\mathbb{P}\left(A(t)=i\right)\right] \notag \\
            &=\mathbb{E}_{\tau(B)}\left[\tau(B)\right]c_{I^*} + \mathbb{E}_{\tau(B)}\left[\sum_{\delta_i>0}\delta_i\sum_{t=1}^{\tau(B)}\mathbb{P}\left(A(t)=i\right)\right]\notag\\
            &=\mathbb{E}_{\tau(B)}\left[\tau(B)\right]c_{I^*} + \mathbb{E}_{\tau(B)}\left[\sum_{\delta_i>0}\delta_i\cdot\mathbb{E}\left[T_i(\tau(B))\vert \tau(B)\right]\right]\notag\\
            &\leq\mathbb{E}_{\tau(B)}\left[\tau(B)\right]c_{I^*} + \mathbb{E}_{\tau(B)}\left[\sum_{\delta_i>0}\delta_i\left(\scriptstyle \frac{8\log(\tau(B))}{Mc_i^2\Delta_i^2} + 1 + \frac{\pi^{2}}{3}\right)\right]\notag\\
            &\leq\mathbb{E}_{\tau(B)}\left[\tau(B)\right]c_{I^*} + \mathbb{E}_{\tau(B)}\left[\sum_{\delta_i>0}\delta_i\left(\scriptstyle \frac{8\log\left(\frac{B}{Mc_{min}}\right)}{Mc_i^2\Delta_i^2} + 1 + \frac{\pi^{2}}{3}\right)\right]\notag\\
            &\leq\mathbb{E}_{\tau(B)}\left[\tau(B)\right]c_{I^*} + \sum_{\delta_i>0}\delta_i\left(\scriptstyle \frac{8\log\left(\frac{B}{Mc_{min}}\right)}{Mc_i^2\Delta_i^2} + 1 + \frac{\pi^{2}}{3}\right).\notag
        \end{align}
        By dividing $c_{I^*}$ on both sides of the inequality above as well as moving items, we can obtain the result:
        \begin{align}
            \mathbb{E}_{\tau(B)}\left[\tau(B)\right]&\geq\frac{B}{Mc_{I^*}} - \frac{c_{min}}{c_{I^*}} \notag\\
            &-\sum_{\delta_i>0}\frac{\delta_i}{c_{I^*}}\left(\frac{8\log\left(\frac{B}{Mc_{min}}\right)}{Mc_i^2\Delta_i^2} + 1 + \frac{\pi^{2}}{3}\right). \notag
        \end{align}
        Thus, this theorem has been proven.
    \end{proof}

    % \begin{lem}[Jensen's inequality]\label{jensen}
    %     For a concave function $g(\cdot):\ \mathbb{D}\to \mathbb{R}$, where $\mathbb{D}$ is an interval in $\mathbb{R}$. If $x_1, x_1, \cdots, x_n\in \mathbb{D}$ and non-negative real number $q_i$ such that $q_1+q_2+\cdots+q_n=1$, we have:
    %     \begin{equation}
    %         q_1g(x_1)+\cdots+q_ng(x_n) \leq g(q_1x_1+\cdots+q_nx_n). \notag
    %     \end{equation}
    %     Furthermore, if $X$ is a random variable, then:
    %     \begin{equation}
    %         \mathbb{E}[g(X)] \leq g\left(\mathbb{E}\left[X\right]\right). \notag
    %     \end{equation}
    % \end{lem}

\begin{figure*}[!t]
	\centering
    \includegraphics[width=1\linewidth]{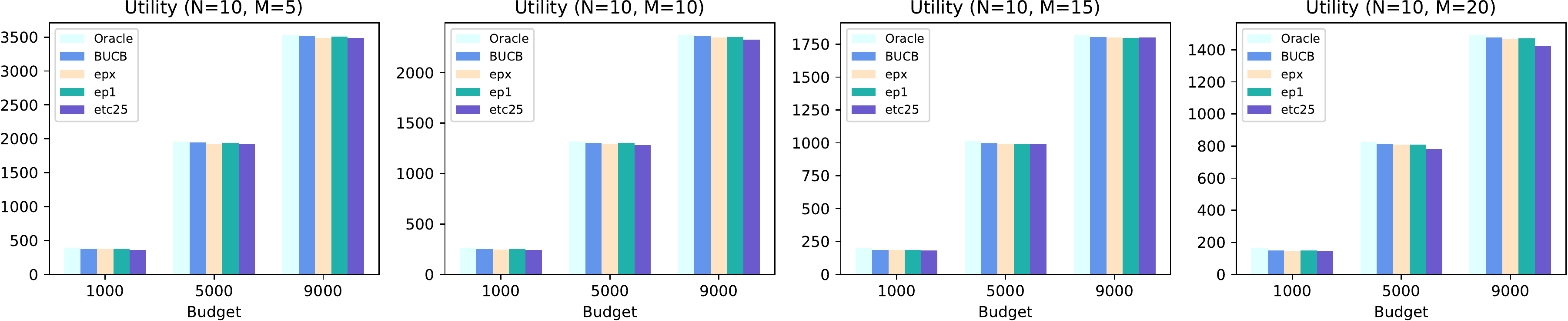}
    %\caption{fig_c2_reward}
	\caption{The cumulative utility versus budget $B$ with $N=10$.}
	\label{fig_reward_N10}
\end{figure*}
\begin{figure*}[!t]
	\centering
    \includegraphics[width=1\linewidth]{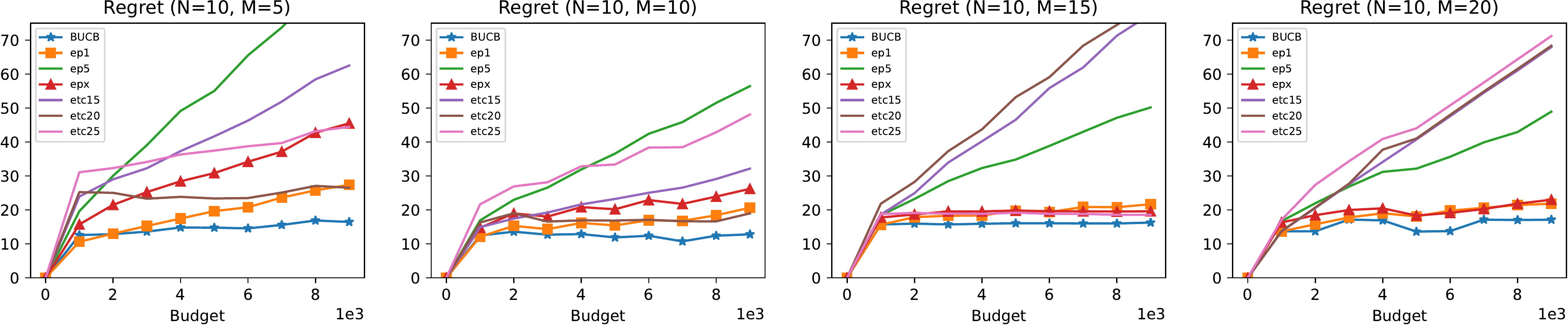}
    %\caption{fig_c2_regret}
	\caption{The regret versus budget $B$ with $N=10$.}
	\label{fig_regret_N10}
\end{figure*}
\begin{figure*}[!t]
	\centering
    \includegraphics[width=1\linewidth]{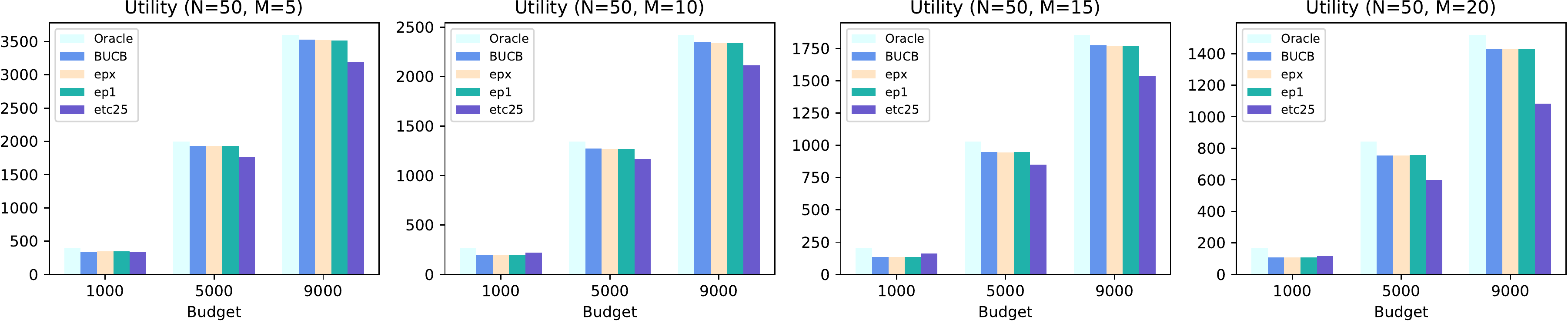}
    %\caption{fig1}
	\caption{The cumulative utility versus budget $B$ with $N=50$.}
	\label{fig_reward_N50}
\end{figure*}
\begin{figure*}[!t]
	\centering
    \includegraphics[width=1\linewidth]{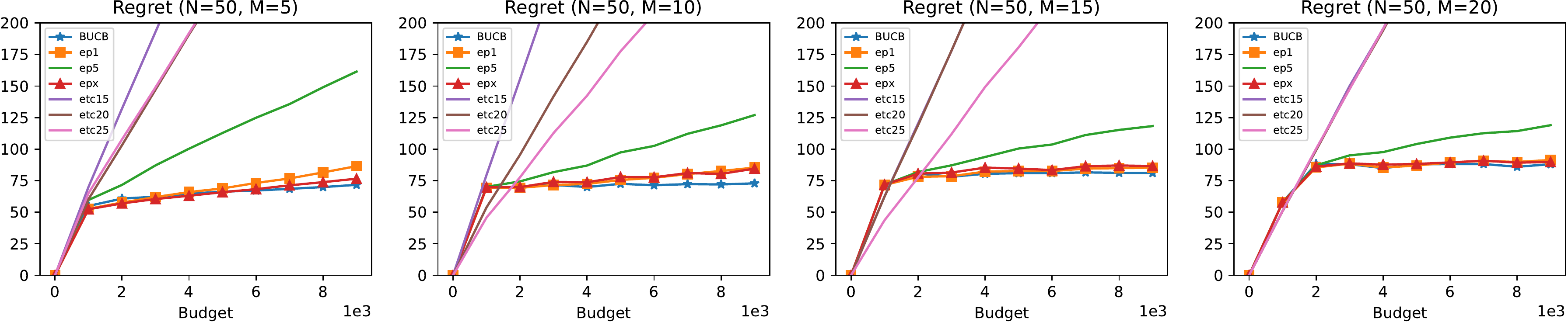}
    %\caption{fig_regret_N50}
	\caption{The regret versus budget $B$ with $N=50$.}
	\label{fig_regret_N50}
\end{figure*}

\begin{lem}[Jensen's inequality]\label{jensen}
    For a concave function $g(\cdot):\ \mathbb{D}\to \mathbb{R}$, where $\mathbb{D}$ is an interval in $\mathbb{R}$. If $x_1, x_1, \cdots, x_n\in \mathbb{D}$ and non-negative real number $q_i$ such that $q_1+q_2+\cdots+q_n=1$, we have \cite{jensen1906fonctions}:
    \begin{equation}
        q_1g(x_1)+q_2g(x_2)+\cdots+q_ng(x_n) \leq g(q_1x_1+q_2x_2+\cdots+q_nx_n). \notag
    \end{equation}
    Furthermore, if $X$ is a random variable, then:
    \begin{equation}
        \mathbb{E}[g(X)] \leq g\left(\mathbb{E}\left[X\right]\right).\notag
    \end{equation}
\end{lem}
Based on the above analysis, we can get the following main theorem, which indicates that the expected regret of our proposed online algorithm is sublinear with respect to the budget $B$. It is important to guarantee the efficiency of our algorithm. 

    \begin{thm}\label{regret_bound}
        The expected regret, denoted by $\mathbb{E}\left[Reg_B\left(A\right)\right]$, can be upper bounded by $\mathbb{E}\left[Reg_B\left(A\right)\right]\leq$
        \begin{align}
            &\log\left(MU_{I^*}\right)\left(1+\frac{c_{min}}{c_{I^*}} + \sum_{\delta_i>0}\frac{\delta_i}{c_{I^*}}\left(\frac{8\log\left(\frac{B}{Mc_{min}}\right)}{Mc_i^2\Delta_i^2}+ 1 \right.\right.\notag\\
            &\left.\left.+ \frac{\pi^{2}}{3}\right)\right) + \log\left(\frac{U_{I^*}}{u_{min}}\right)N\left(\frac{8\log\left(\frac{B}{Mc_{min}}\right)}{Mc_i^2\Delta_i^2} + 1 + \frac{\pi^{2}}{3}\right). \notag
        \end{align}
    \end{thm}
    \begin{proof}
        The proof of Theorem \ref{regret_bound} is based on Theorem \ref{Ti_bound}, Theorem \ref{tau_bound}, and Lemma \ref{jensen}. Due to the space constraint, it is relegated to Appendix \ref{regret_proof}.
    \end{proof}

\begin{figure*}[!t]
	\centering
    \includegraphics[width=1\linewidth]{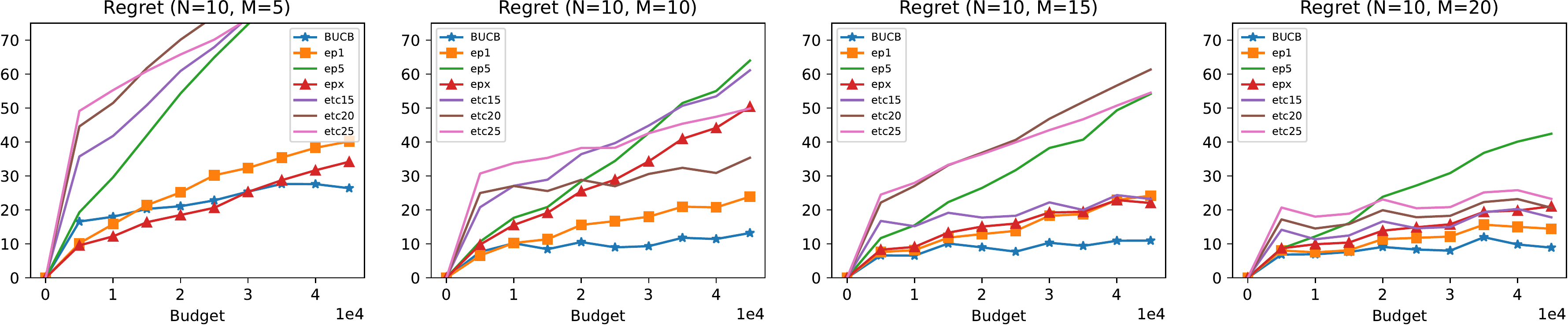}
    %\caption{fig_c10_reward}
	\caption{The regret versus budget $B$ with low variation in cost.}
	\label{fig_regret_N10_low}
\end{figure*}

\begin{figure*}[!t]
	\centering
    \includegraphics[width=1\linewidth]{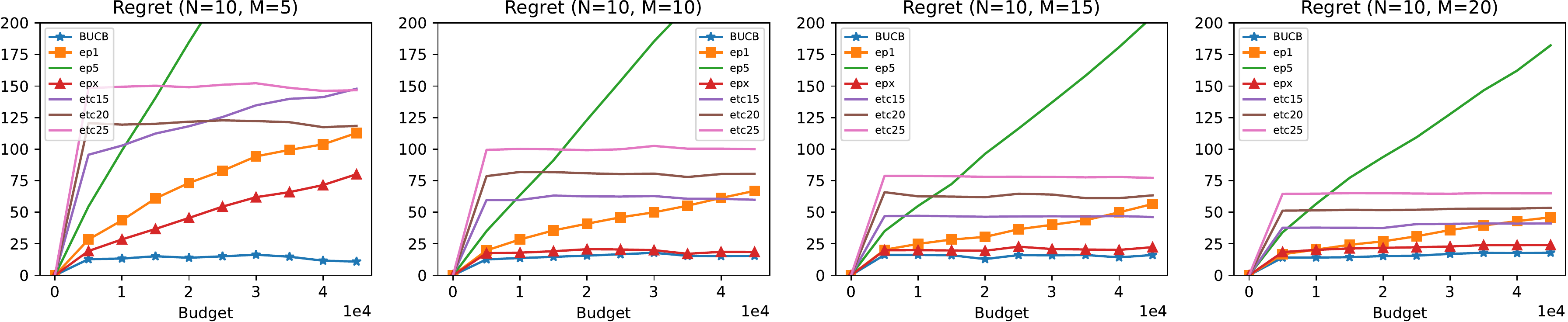}
    %\caption{fig_regret_N10_medium}
	\caption{The regret versus budget $B$ with medium variation in cost.}
	\label{fig_regret_N10_medium}
\end{figure*}

\begin{figure*}[!t]
	\centering
    \includegraphics[width=1\linewidth]{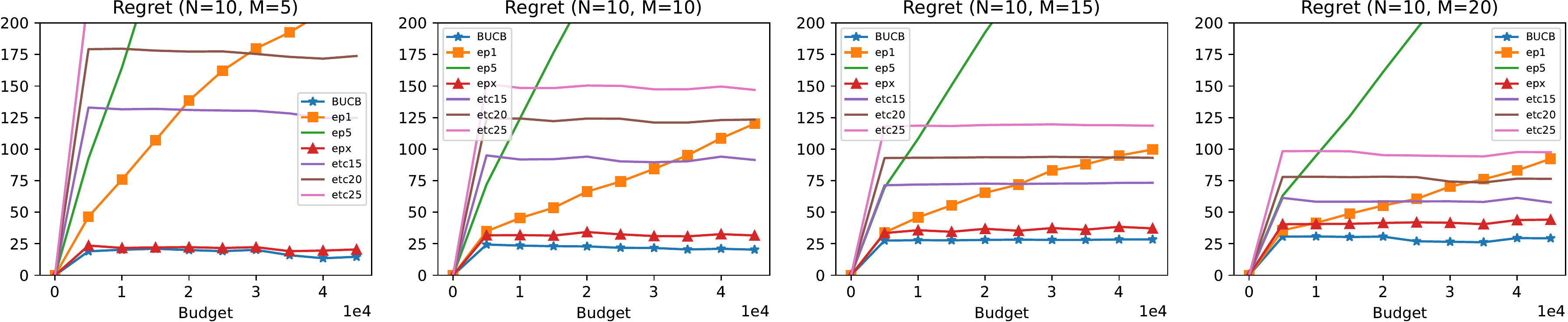}
    %\caption{fig_c10_regret}
	\caption{The regret versus budget $B$ with high variation in cost.}
	\label{fig_regret_N10_high}
\end{figure*}

    \begin{figure*}[t]
    	\centering
        \includegraphics[width=1\linewidth]{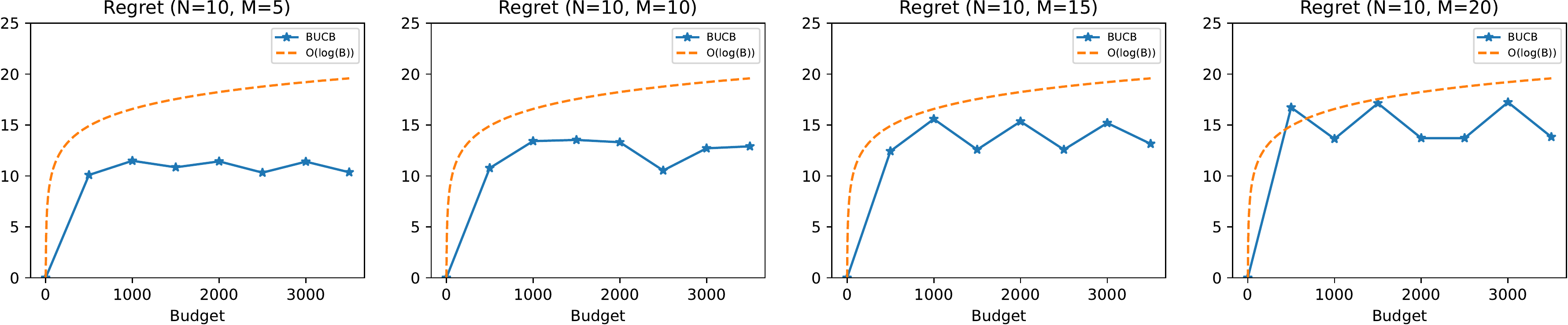}
        %\caption{fig1}
    	\caption{The complexity of Budgeted-UCB regret versus budget $B$ in contrast to $\mathcal{O}(\log(B))$ with $N=10$.}
    	\label{fig_log_N10}
    \end{figure*}

    \begin{figure*}[t]
    	\centering
        \includegraphics[width=1\linewidth]{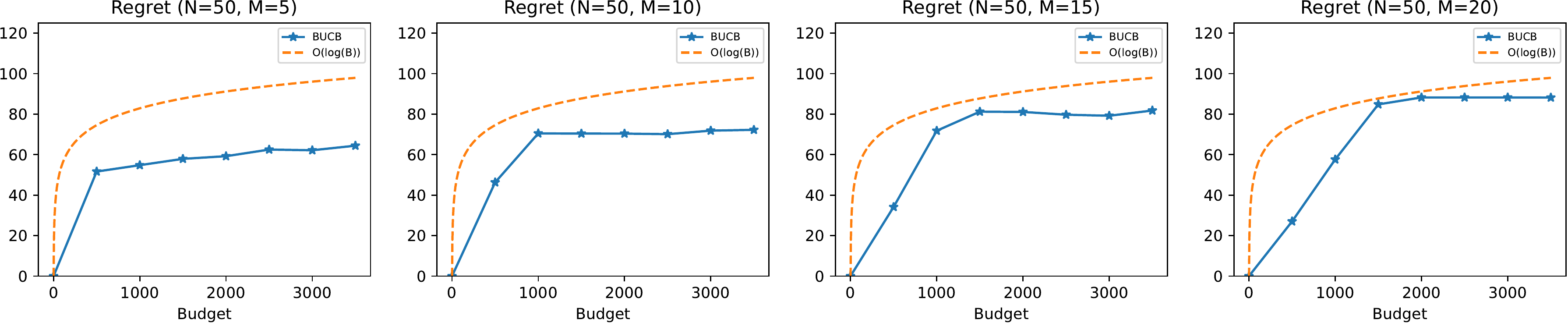}
        %\caption{fig1}
    	\caption{The complexity of Budgeted-UCB regret versus budget $B$ in contrast to $\mathcal{O}(\log(B))$ with $N=10$.}
    	\label{fig_log_N50}
    \end{figure*}

\section{Performance Evaluation}\label{sec6}
    In this section, we evaluate the performance of the proposed algorithm by comparing it with other commonly used baselines. All of our simulations are programmed by Python 3.9 and run on a Windows 10 platform with CPU of Intel Core i7-11700 and 32GB RAM. The source code of our implementation can be shown in the following link: \url{https://github.com/Chubbro/BUCB}.

\subsection{Experimental Setup}
    The distribution behind each MC is simulated by a truncated Gaussian with the expected QoE $U_i$, which truncates an original Gaussian distribution to the range $[1,2]$. %The mean of the original Gaussian distribution is generated randomly from a uniform distribution between 1 and 2.
    At the same time, the unit cost $c_i$ of each MC $i\in\mathcal{N}$ is taken from a uniform distribution from the interval $[1,2]$. We first demonstrate the relative relation and the change of the cumulative utility and algorithmic regrets with the budget for the proposed and the benchmark algorithms, as $M$ varies from 5 to 20 and $N$ is set to 10 and 50. Next, for mimicking the real-world situations of low, medium, and high variance in costs, we take three groups of experiments with different intervals of the costs as $[1.5,2]$, $[1,2]$, and $[1,3]$, respectively. Within each group, the costs of the MCs are all sampled from the uniform distribution with one of the three intervals. For convenience, we label the three groups as \textit{low}, \textit{mediun}, and \textit{high} variations in cost in the remaining of this paper. In addition, we illustrate the logarithmic complexity of the proposed algorithm in terms of the budget. As aforementioned, the users served by the same MC receive QoE from a common distribution related to this MC, hence the feedback of QoE for the users can be sampled from the same truncated Gaussian distribution.

\subsection{Benchmarks}
    In the experiments, to demonstrate the efficacy, our proposed algorithm is compared with the following benchmarks:
    \begin{itemize}
        \item \textit{Oracle}: the Oracle algorithm has the prior knowledge of real expected QoE for each MC. In every round, the oracle selects the MC with the greatest expected utility-cost ratio.
        \item \textit{Budgeted-UCB (BUCB)}: The proposed algorithm shown as Algorithm \ref{BUCB} in this paper.
        \item \textit{ep1}, \textit{ep5}, \textit{epx}: These three algorithms are of classic $\varepsilon$-greedy strategy which selects the empirically best MC with probability $1-\varepsilon$ or explores a random MC with probability $\varepsilon$, with an only difference in $\varepsilon$ as 0.01, 0.05 and $1/t$ respectively.
        \item \textit{etc15}, \textit{etc20}, \textit{etc25}: These are variants of the explore-then-commit (ETC) algorithm, also known as epsilon-first strategy, with budget constraints and different proportions ($15\%$, $20\%$ and $25\%$) of the total budget for exploration. The algorithm would not conduct the exploitation phase until a certain proportion of the budget is exhausted through exploring all the MCs successively and repeatedly. Then based on the data collected in the exploration phase, the algorithm always selects the empirically best MC during the exploitation phase.
    \end{itemize}

\subsection{Result Analysis}
    We investigate the variation of utility and regret against budgets under different combinations of the number of candidate MCs $N$ and the number of connected users $M$, as shown from Fig. \ref{fig_reward_N10} to Fig. \ref{fig_regret_N50}, respectively. For reward comparison in Fig. \ref{fig_reward_N10} and Fig. \ref{fig_reward_N50}, on account of space limitations, the distinctly inferior algorithms compared with others from the results, including \textit{ep5}, \textit{etc20}, and \textit{etc25}, are omitted. It can be seen that the proposed \textit{BUCB} outperforms all the benchmarks in all the cases. However, as $M$ increases, the differences in the performance of algorithms are becoming minor referring to Fig. \ref{fig_regret_N10} and Fig. \ref{fig_regret_N50}, which is due to the increasingly accurate estimation brought by a larger sampling number $M$. Thus \textit{BUCB} has a more prominent advantage over others when $M=5$ and $M=10$ in contrast to the cases when $M=15,20$. From the results, the $\varepsilon$-greedy based algorithms are superior to the ETC-based algorithms on a whole. Also, for $N=10$, \textit{BUCB} has clearly better regrets than \textit{ep1} and \textit{epx}, while this ascendancy attenuates as $N$ increases to 50 despite the scaling. This is because the augmented number of arms would bring an extra amount of suboptimal exploration, which leads to the relatively small advantage of \textit{BUCB} under the same budget compared with the situations of fewer arms. 
    
    Furthermore, we model the real cases of different variations in cost as groups of low, medium, and high variations, and compare the regret of the proposed algorithm with other benchmarks in each group, as shown in Fig. \ref{fig_regret_N10_low}, Fig. \ref{fig_regret_N10_medium}, and Fig. \ref{fig_regret_N10_high}. Here we extend the budget bound up to 50000 for revealing clearer trends. In all three groups, \textit{BUCB} still shows better performance over the $\varepsilon$-greedy based algorithms and the ETC-based algorithms. It is interesting that the subtle differences in costs of the low-variation group contribute to the similarities in the utility-cost ratios, leading to slow convergence as depicted in Fig. \ref{fig_regret_N10_low} in contrast to the relatively rapid convergence in Fig. \ref{fig_regret_N10_medium} and Fig. \ref{fig_regret_N10_high}. For example, in Fig. \ref{fig_regret_N10_low}, we can see all the algorithms end with growing regrets, while in most cases of Fig. \ref{fig_regret_N10_medium} and Fig. \ref{fig_regret_N10_high}, \textit{BUCB}, \textit{epx}, and the ETC-based algorithms reach relative flat regret growth due to the relative sufficient exploration.

    % \begin{figure}[!t]
    % 	\centering
    %         \includegraphics[width=\linewidth]{graph/c5_c10_log.pdf}
    % 	\caption{The budgeted UCB complexity of regret versus budget $B$ in contrast to $\mathcal{O}(\log(B))$ when $N=100$ and $M=100$.}
    % 	\label{fig_log_c5_c10}
    % \end{figure}

    To practically demonstrate the logarithmic complexity for the regret of \textit{BUCB}, we also compare the regret's growth rates between \textit{BUCB} and $\mathcal{O}(\log(B))$ function as shown in Fig. \ref{fig_log_N10} and Fig. \ref{fig_log_N50}, which corresponds to $N=10$ and $N=50$, respectively. The contrasted logarithmic functions in those figures are of the form $m\cdot\log(B)$ in which the coefficient $m$ is set no more than $12$. It is illustrated that the proposed algorithm practically achieves regrets with $\mathcal{O}(\log(B))$ complexity in regard to budget $B$.

\section{Conclusion}\label{sec7}
    This paper introduces Meta Computing as a novel computing paradigm that aims to integrate all the resources on a network as an MC. While technical limitations currently prevent the integration of resources across the entire network, we propose a meta computing system architecture, composed of multiple MCs in small-scale networks, which presents a viable solution. We highlight the importance of considering user experience in meta computing systems, particularly in scheduling MCs to maximize QoE with limited budgets. To address the resource scheduling problem, we propose a UCB-based algorithm, formulated from the MAB perspective, and model the utility of service using a concave function of total QoE to capture the ubiquitous law of diminishing marginal utility. Theoretical analysis demonstrates an upper bound on the regret of the proposed algorithm with logarithmic growth to the budget. Extensive experiments validate the correctness and effectiveness of the algorithm. 
    
    The research presented in this paper contributes to the advancement of Meta Computing as a promising computing paradigm and provides insights into addressing the challenges of resource scheduling. Future work could further explore other algorithms or techniques, taking into consideration practical constraints and real-world scenarios, to further advance the field of Meta Computing and its wide application.

\section*{Acknowledgment}

This work was supported in part by the National Key R\&D Program of China under Grant No. 2022YFE0201400, the National Natural Science Foundation of China (NSFC) under Grant No. 62202055, the Start-up Fund from Beijing Normal University under Grant No. 310432104, the Start-up Fund from BNU-HKBU United International College under Grant No. UICR0700018-22, and the Project of Young Innovative Talents of Guangdong Education Department under Grant No. 2022KQNCX102.

\ifCLASSOPTIONcaptionsoff
  \newpage
\fi

\bibliographystyle{IEEEtran}
\bibliography{references}

\begin{IEEEbiography}[{\includegraphics[width=1in,height=1.25in,clip,keepaspectratio]{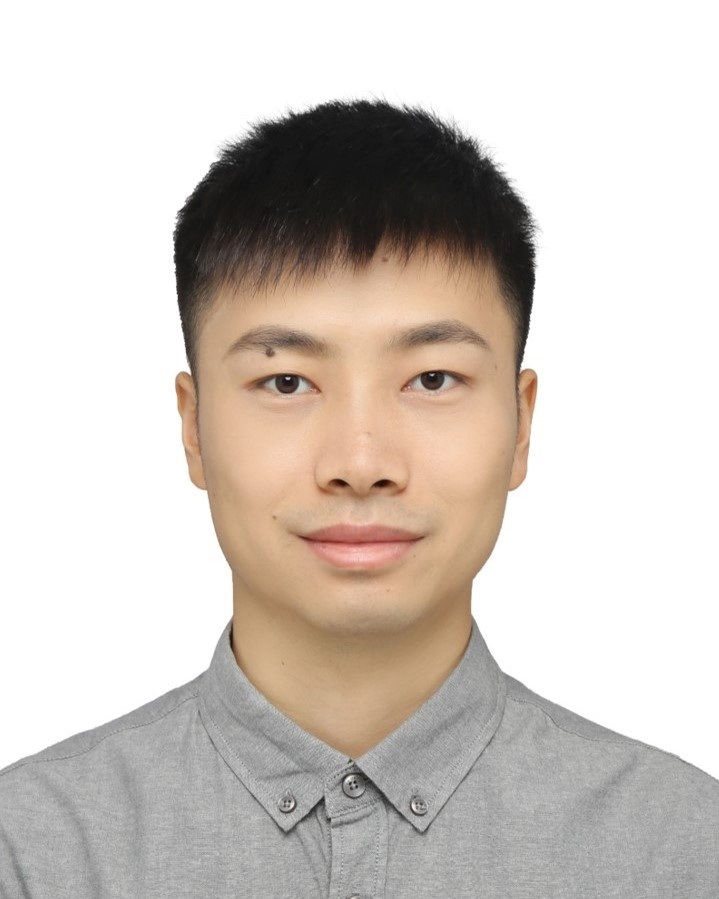}}]{Yandi Li} received his B.E. degree from the School of Optoelectronic Science and Engineering, University of Electronic Science and Technology of China, Chengdu, China, in 2013. He is currently pursuing the M.Phil. degree with the Guangdong Key Lab of AI and Multi-Modal Data Processing, Department of Computer Science, BNU-HKBU United International College, Zhuhai, China. He is supervised by Dr. Jianxiong Guo, and his research interests include social networks, online algorithms, and machine learning.
\end{IEEEbiography}

\begin{IEEEbiography}[{\includegraphics[width=1in,height=1.25in,clip,keepaspectratio]{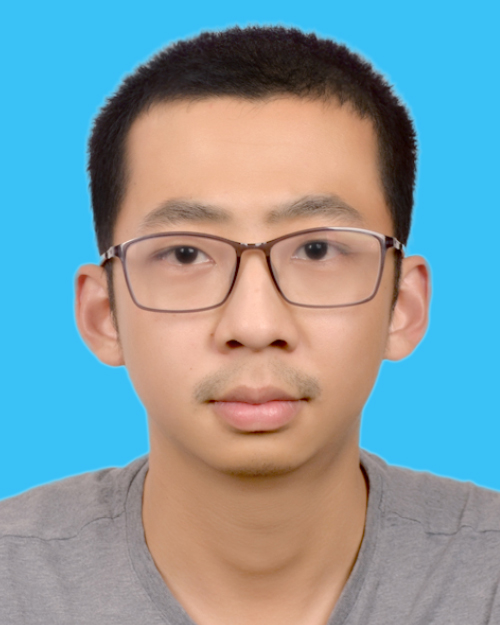}}]{Jianxiong Guo}
    (Member, IEEE) received his Ph.D. degree from the Department of Computer Science, University of Texas at Dallas, Richardson, TX, USA, in 2021, and his B.E. degree from the School of Chemistry and Chemical Engineering, South China University of Technology, Guangzhou, China, in 2015. He is currently an Assistant Professor with the Advanced Institute of Natural Sciences, Beijing Normal University, and also with the Guangdong Key Lab of AI and Multi-Modal Data Processing, BNU-HKBU United International College, Zhuhai, China. He is a member of IEEE/ACM/CCF. He has published more than 40 papers and been a reviewer in famous international journals/conferences. His research interests include social networks, algorithm design, data mining, IoT application, blockchain, and combinatorial optimization.
\end{IEEEbiography}

\begin{IEEEbiography}[{\includegraphics[width=1in,height=1.25in,clip,keepaspectratio]{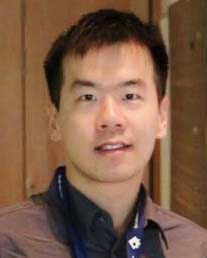}}]{Yupeng Li} (Member, IEEE) received the Ph.D. degree in computer science from The University of Hong Kong. He was with the University of Toronto and is currently with Hong Kong Baptist University. His research interests are in general areas of network science and, in particular, algorithmic decision making and machine learning problems, which arise in networked systems. %, such as information networks and ride-sharing platforms. 
He is also excited about interdisciplinary research that applies algorithmic techniques to edging problems. Recently, he has worked on robust online machine learning for the application of data classification, and he has extended these techniques to modern areas in networking and social media. Dr. Li has been awarded the Rising Star in Social Computing Award by CAAI and the distinction of Distinguished Member of the IEEE INFOCOM Technical Program Committee in 2022. He serves on the technical committees of some top conferences in computer science. His works have been published in prestigious venues, such as \textsc{IEEE INFOCOM}, \textsc{ACM MobiHoc}, \textsc{IEEE Journal on Selected Areas in Communications}, and \textsc{IEEE/ACM Transactions on Networking}. He is a member of ACM and IEEE.
\end{IEEEbiography}

\begin{IEEEbiography}[{\includegraphics[width=1in,height=1.25in,clip,keepaspectratio]{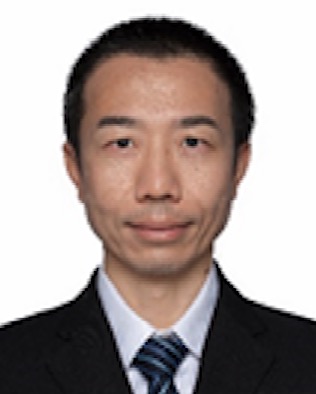}}]{Tian Wang}
	received his BSc and MSc degrees in Computer Science from the Central South University in 2004 and 2007, respectively. He received his PhD degree in City University of Hong Kong in Computer Science in 2011. Currently, he is a professor in the Institute of Artificial Intelligence and Future Networks, Beijing Normal University \& UIC. His research interests include internet of things, edge computing and mobile computing. He has 27 patents and has published more than 200 papers in high-level journals and conferences. He has more than 11000 citations, according to Google Scholar. His H-index is 53. He has managed 6 national natural science projects (including 2 sub-projects) and 4 provincial-level projects.
\end{IEEEbiography}

\begin{IEEEbiography}[{\includegraphics[width=1in,height=1.25in,clip,keepaspectratio]{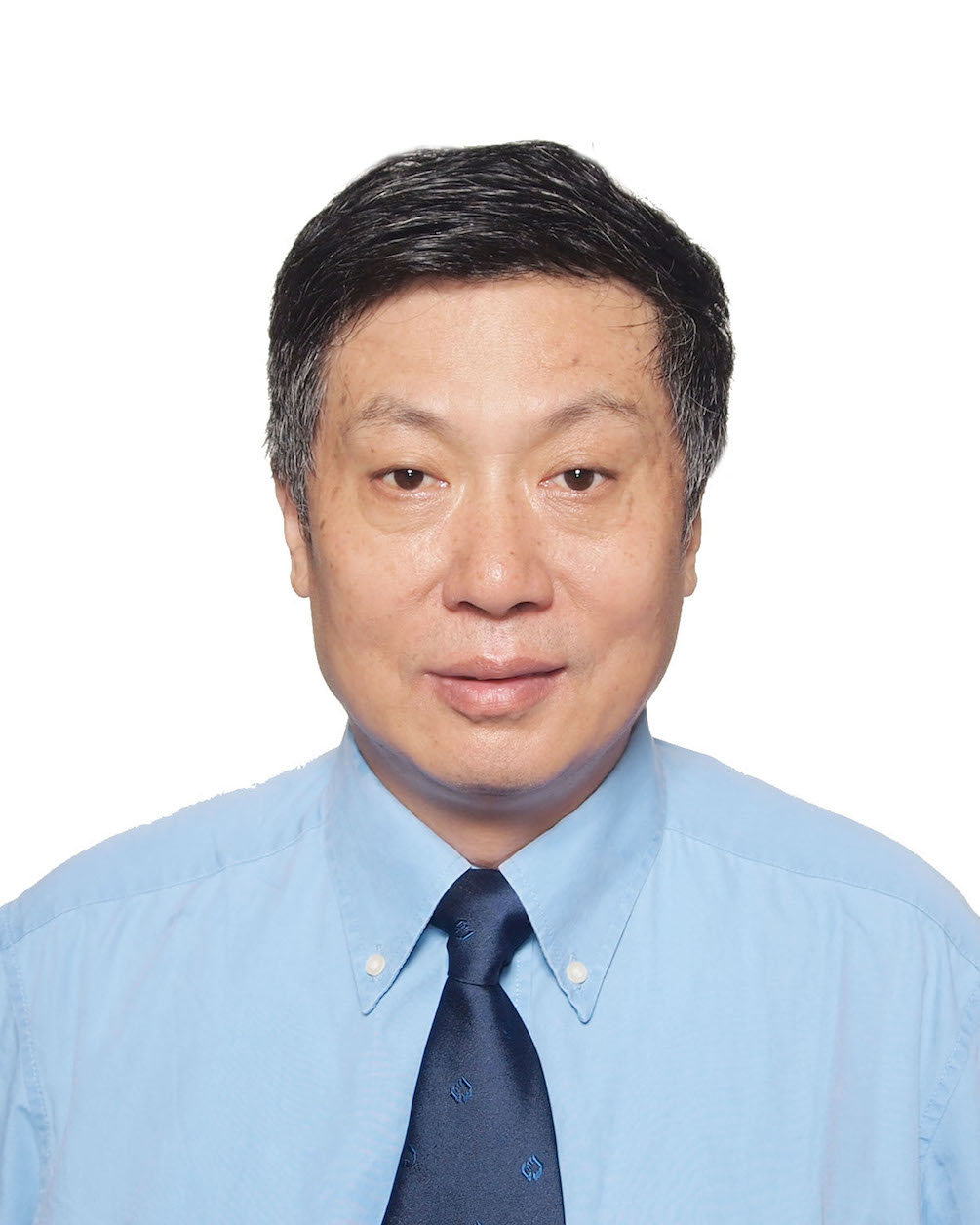}}]{Weijia Jia} (Fellow, IEEE)
    is currently a Chair Professor, Director of BNU-UIC Institute of Artificial Intelligence and Future Networks, Beijing Normal University (Zhuhai) and VP for Research of BNU-HKBU United International College (UIC) and has been the Zhiyuan Chair Professor of Shanghai Jiao Tong University, China. He was the Chair Professor and the Deputy Director of State Kay Laboratory of Internet of Things for Smart City at the University of Macau. He received BSc/MSc from Center South University, China in 82/84 and Master of Applied Sci./PhD from Polytechnic Faculty of Mons, Belgium in 92/93, respectively, all in computer science. From 93-95, he joined German National Research Center for Information Science (GMD) in Bonn (St. Augustine) as a research fellow. From 95-13, he worked in City University of Hong Kong as a professor. His contributions have been recognized as optimal network routing and deployment; anycast and QoS routing, sensors networking, AI (knowledge relation extractions; NLP, etc.), and edge computing. He has over 600 publications in the prestige international journals/conferences and research books and book chapters. He has received the best product awards from International Science \& Tech. Expo (Shenzhen) in 20112012 and the 1st Prize of Scientific Research Awards from the Ministry of Education of China in 2017 (list 2). He has served as area editor for various prestige international journals, chair, and PC member/keynote speaker for many top international conferences. He is the Fellow of IEEE and the Distinguished Member of CCF.
\end{IEEEbiography}

\onecolumn
\appendix
\setcounter{thm}{2}
\setcounter{lem}{0}
\renewcommand{\thelem}{A\arabic{lem}}
\renewcommand{\thesection}{\Alph{section}.\arabic{section}}
\setcounter{section}{0}

\subsection{Proof of Theorem \ref{regret_bound}} \label{regret_proof}
\begin{thm}
    The expected regret can be upper bounded as $\mathbb{E}\left[Reg_B\left(A\right)\right] \leq$
    \begin{equation}
        \log\left(MU_{I^*}\right)\left(1+\frac{c_{min}}{c_{I^*}} + \sum_{\delta_i>0}\frac{\delta_i}{c_{I^*}}\left(\frac{8\log\left(\frac{B}{Mc_{min}}\right)}{Mc_i^2\Delta_i^2} + 1 + \frac{\pi^{2}}{3}\right)\right) + \log\left(\frac{U_{I^*}}{u_min}\right)N\left(\frac{8\log\left(\frac{B}{Mc_{min}}\right)}{Mc_i^2\Delta_i^2} + 1 + \frac{\pi^{2}}{3}\right). \notag
    \end{equation}
\end{thm}

\begin{proof}
    First, we transform $\mathbb{E}_{\tau(B),\{A(t)\}}\left[R_B\left(A\right)\right]$ as follows.
    \begin{align}
        \mathbb{E}_{\tau(B),\{A(t)\}}\left[R_B\left(A\right)\right] 
        &= \mathbb{E}_{\tau(B),\{A(t)\}}\left[\sum_{t=1}^{\tau(B)}\log\left(\sum_{j=1}^{M}u_{A(t)j}\left(t\right)\right)\right] \notag \\
        &= \mathbb{E}_{\tau(B),\{A(t)\}}\left[\sum_{t=1}^{\tau(B)}\log\left(\sum_{j=1}^{M}M\frac{u_{A(t)j}\left(t\right)}{M}\right)\right] \notag\\
        &= \mathbb{E}_{\tau(B),\{A(t)\}}\left[\sum_{t=1}^{\tau(B)}\left(\log(M) + \log\left(\sum_{j=1}^{M}\frac{u_{A(t)j}\left(t\right)}{M}\right)\right)\right] \notag\\
        &\geq \mathbb{E}_{\tau(B),\{A(t)\}}\left[\sum_{t=1}^{\tau(B)}\left(\log(M) + \frac{1}{M}\sum_{j=1}^{M}\log(u_{A(t)j}(t))\right)\right] \qquad \triangleright \; \text{Lemma \ref{jensen}.}\notag\\
        &= \mathbb{E}_{\tau(B)}\mathbb{E}_{\{A(t)\}\vert\tau(B)}\left[\sum_{t=1}^{\tau(B)}\left(\log(M) + \frac{1}{M}\sum_{j=1}^{M}\log(u_{A(t)j}(t))\right)\right] \notag\\ 
        &= \mathbb{E}_{\tau(B)}\left[\sum_{t=1}^{\tau(B)}\mathbb{E}_{A(t)}\left[\log(M) + \frac{1}{M}\sum_{j=1}^{M}\log(u_{A(t)j}(t))\right]\right] \notag\\
       &= \mathbb{E}_{\tau(B)}\left[\sum_{t=1}^{\tau(B)}\sum_{i=1}^{N}\left(\mathbb{P}\left(A\left(t\right)=i\right)\left(\log(M) + \frac{1}{M}\sum_{j=1}^{M}\log(u_{ij}(t))\right)\right)\right] \notag\\
        &= \mathbb{E}_{\tau(B)}\left[\sum_{t=1}^{\tau(B)}\sum_{i=1}^{N}\mathbb{P}\left(A\left(t\right)=i\right)\log(M) + \sum_{t=1}^{\tau(B)}\sum_{i=1}^{N}\mathbb{P}\left(A\left(t\right)=i\right)\frac{1}{M}\sum_{j=1}^{M}\log(u_{ij}(t))\right] \notag\\
        &= \mathbb{E}_{\tau(B)}\left[\sum_{t=1}^{\tau(B)}\log(M) + \sum_{t=1}^{\tau(B)}\sum_{i=1}^{N}\mathbb{P}\left(A\left(t\right)=i\right)\frac{1}{M}\sum_{j=1}^{M}\log(u_{ij}(t))\right] \notag\\
        &= \mathbb{E}_{\tau(B)}\left[\tau(B)\log(M)\right] + \mathbb{E}_{\tau(B)}\left[\sum_{t=1}^{\tau(B)}\sum_{i=1}^{N}\mathbb{P}\left(A\left(t\right)=i\right)\frac{1}{M}\sum_{j=1}^{M}\log(u_{ij}(t))\right] \notag\\ 
        &= \log(M)\mathbb{E}_{\tau(B)}\left[\tau(B)\right] + \mathbb{E}_{\tau(B)}\left[\sum_{t=1}^{\tau(B)}\sum_{i=1}^{N}\mathbb{P}\left(A\left(t\right)=i\right)\frac{1}{M}\sum_{j=1}^{M}\log(u_{ij}(t))\right]. \notag
    \end{align}
    \noindent
    Second, we transform
    $\mathbb{E}\left[R_B\left(A^*\right)\right]$ as follows.
    \begin{align}
        \mathbb{E}\left[R_B\left(A^*\right)\right]
        &\leq \mathbb{E}\left[\sum_{t=1}^{\lfloor B/(Mc_{I^*})\rfloor+1}\log\left(\sum_{j=1}^{M}u_{I^*j}\left(t\right)\right)\right] \notag\\
        &= \sum_{t=1}^{\lfloor B/(Mc_{I^*})\rfloor+1}\mathbb{E}\left[\log\left(\sum_{j=1}^{M}u_{I^*j}(t)\right)\right] \notag\\
        &\leq \sum_{t=1}^{\lfloor B/(Mc_{I^*})\rfloor+1}\log\left(\mathbb{E}\left[\sum_{j=1}^{M}u_{I^*j}(t)\right]\right)=\sum_{t=1}^{\lfloor B/(Mc_{I^*})\rfloor+1}\log\left(\sum_{j=1}^{M}\mathbb{E}[u_{I^*j}(t)]\right) \qquad \triangleright \; \text{Lemma \ref{jensen}.}\notag\\
        &= \sum_{t=1}^{\lfloor B/(Mc_{I^*})\rfloor+1}\log\left(M\cdot U_{I^*}\right). \notag
    \end{align}
\noindent
For convenience, we define
\begin{equation}
    \mathcal{Q}[\tau(B)]=\mathbb{E}_{\tau(B)}\left[\sum_{t=1}^{\tau(B)}\sum_{i=1}^{N}\mathbb{P}\left(A\left(t\right)=i\right)\frac{1}{M}\sum_{j=1}^{M}\log(u_{ij}(t))\right].\notag
\end{equation}
Finally, we have
    \begin{flalign}
        &\mathbb{E}_{\tau(B),\{A(t)\}}\left[Reg_A\left(B\right)\right]= \mathbb{E}\left[R_B\left(A^*\right)\right] - \mathbb{E}_{\tau(B),\{A(t)\}}\left[R_B\left(A\right)\right]\notag\\
        &\leq \sum_{t=1}^{\lfloor B/(Mc_{I^*})\rfloor+1}\log\left(MU_{I^*}\right) - \log(M)\mathbb{E}\left[\tau(B)\right] - \mathcal{Q}[\tau(B)] \notag\\
        &= \sum_{t=1}^{\lfloor B/(Mc_{I^*})\rfloor+1}\log\left(MU_{I^*}\right) - \log(M)\mathbb{E}\left[\tau(B)\right] - \log(U_{I^*})\mathbb{E}\left[\tau(B)\right] + \log(U_{I^*})\mathbb{E}\left[\tau(B)\right]- \mathcal{Q}[\tau(B)] \notag\\
        &= \sum_{t=1}^{\lfloor B/(Mc_{I^*})\rfloor+1}\log\left(MU_{I^*}\right) - \log(M)\mathbb{E}\left[\tau(B)\right] - \log(U_{I^*})\mathbb{E}\left[\tau(B)\right] + \mathbb{E}\left[\sum_{t=1}^{\tau(B)}\log(U_{I^*})\right]-\mathcal{Q}[\tau(B)] \notag\\
        &= \sum_{t=1}^{\lfloor B/(Mc_{I^*})\rfloor+1}\log\left(MU_{I^*}\right) - \log(MU_{I^*})\mathbb{E}\left[\tau(B)\right] + \mathbb{E}\left[\sum_{t=1}^{\tau(B)}\log(U_{I^*})\sum_{i=1}^{N}\mathbb{P}\left(A\left(t\right)=i\right)\right]-\mathcal{Q}[\tau(B)] \notag\\
        &= \sum_{t=1}^{\lfloor B/(Mc_{I^*})\rfloor+1}\log\left(MU_{I^*}\right) - \log(MU_{I^*})\mathbb{E}\left[\tau(B)\right] + \mathbb{E}_{\tau(B)}\left[\sum_{t=1}^{\tau(B)}\sum_{i=1}^{N}\mathbb{P}\left(A\left(t\right)=i\right)\frac{1}{M}\sum_{j=1}^{M}\log(U_{I^*})\right]-\mathcal{Q}[\tau(B)] \notag\\
        &= \sum_{t=1}^{\lfloor B/(Mc_{I^*})\rfloor+1}\log\left(MU_{I^*}\right) - \log(MU_{I^*})\mathbb{E}\left[\tau(B)\right] + \mathbb{E}_{\tau(B)}\left[\sum_{t=1}^{\tau(B)}\sum_{i=1}^{N}\mathbb{P}\left(A\left(t\right)=i\right)\frac{1}{M}\sum_{j=1}^{M}\log\left(\frac{U_{I^*}}{u_{ij}(t)}\right)\right] \notag\\
        &\leq \sum_{t=1}^{\lfloor B/(Mc_{I^*})\rfloor+1}\log\left(MU_{I^*}\right) - \log(MU_{I^*})\mathbb{E}\left[\tau(B)\right] + \mathbb{E}_{\tau(B)}\left[\sum_{t=1}^{\tau(B)}\sum_{i=1}^{N}\mathbb{P}\left(A\left(t\right)=i\right)\frac{1}{M}\sum_{j=1}^{M}\log\left(\frac{U_{I^*}}{u_{min}}\right)\right] \notag\\
        &= \log(MU_{I^*})\left(\lfloor\frac{B}{Mc_{I^*}}\rfloor + 1 - \mathbb{E}\left[\tau(B)\right]\right) + \log\left(\frac{U_{I^*}}{u_{min}}\right)\mathbb{E}_{\tau(B)}\left[\sum_{i=1}^{N}\sum_{t=1}^{\tau(B)}\mathbb{P}\left(A\left(t\right)=i\right)\right] \notag\\
        &\leq \log(MU_{I^*})\left(\lfloor\frac{B}{Mc_{I^*}}\rfloor + 1 - \left(\frac{B}{Mc_{I^*}} - \frac{c_{min}}{c_{I^*}} - \sum_{\delta_i>0}\frac{\delta_i}{c_{I^*}}\left(\frac{8\log\left(\frac{B}{Mc_{min}}\right)}{Mc_i^2\Delta_i^2} + 1 + \frac{\pi^{2}}{3}\right)\right)\right) \notag\\
        &\qquad+ \log\left(\frac{U_{I^*}}{u_{min}}\right)\mathbb{E}_{\tau(B)}\left[\sum_{i=1}^{N}\left(\frac{8\log(\tau(B))}{Mc_i^2\Delta_i^2} + 1 + \frac{\pi^{2}}{3}\right)\right] \qquad \triangleright \; \text{Theorem \ref{Ti_bound}, \ref{tau_bound}.}\notag\\
        &\leq \log(MU_{I^*})\left(1+\frac{c_{min}}{c_{I^*}} + \sum_{\delta_i>0}\frac{\delta_i}{c_{I^*}}\left(\frac{8\log\left(\frac{B}{Mc_{min}}\right)}{Mc_i^2\Delta_i^2} + 1 + \frac{\pi^{2}}{3}\right)\right)\notag\\
        &\qquad+ \log\left(\frac{U_{I^*}}{u_{min}}\right)\mathbb{E}_{\tau(B)}\left[\sum_{i=1}^{N}\left(\frac{8\log\left(\frac{B}{Mc_{min}}\right)}{Mc_i^2\Delta_i^2} + 1 + \frac{\pi^{2}}{3}\right)\right]\notag\\
        &= \log(MU_{I^*})\left(1+\frac{c_{min}}{c_{I^*}} + \sum_{\delta_i>0}\frac{\delta_i}{c_{I^*}}\left(\frac{8\log\left(\frac{B}{Mc_{min}}\right)}{Mc_i^2\Delta_i^2} + 1 + \frac{\pi^{2}}{3}\right)\right) + \log\left(\frac{U_{I^*}}{u_{min}}\right)N\left(\frac{8\log\left(\frac{B}{Mc_{min}}\right)}{Mc_i^2\Delta_i^2} + 1 + \frac{\pi^{2}}{3}\right). \notag
    \end{flalign}
     The last inequality is obtained from $\tau(B)\leq \frac{B}{c_{min}}$. Thus, this theorem has been proven.
\end{proof}
    
\end{document}